\documentclass{amsart}
\usepackage[english]{babel}
\usepackage[latin1]{inputenc}
\usepackage{amscd,amssymb,amsmath,amsfonts,amsthm}

\usepackage{graphicx}
\usepackage{xcolor}

\usepackage[active]{srcltx}

\usepackage[small,nohug,heads=vee]{diagrams}
\usepackage[dvips]{hyperref}

\diagramstyle[labelstyle=\scriptstyle]

\usepackage[dvips]{hyperref}

\usepackage{soul}
\setul{5ex}{0.5ex}

\usepackage{rotating}

\newcommand{\footnoteocultar}[1]{}

\newcommand{\quitar}[1]{
}

\newcommand{\nuevo}[1]{
#1
}

\newcommand{\linero}[1]{
#1
}

\newcommand{\ojoborrar}[1]{
}

\newcommand{\imparocultar}[1]{
}

\newcommand{\corregido}{
}

\def\iets{IETs}

\def\iet{IET}

\def\O{\mathcal{O}}
\def\d{\mathcal{\delta}}

\def\angulon{\alpha_{2,4}(n)}

\newcommand{\bbb}{
b}
\newcommand{\sss}{
s}
\newcommand{\as}{\alpha_{\bbb,\sss+\bbb}}

\def\+{\oplus}
\def\-{\ominus}

\newcommand{\antin}{{|\pi|^{-1}(n)}}

\makeatletter
\def\@seccntformat#1{\csname the#1\endcsname.\quad}
\makeatother

\newtheorem{theorem}{Theorem}

\newtheorem*{mteo}{Main Theorem}
\newtheorem{coro}[theorem]{Corollary}
\newtheorem{lemma}[theorem]{Lemma}
\newtheorem{proposition}[theorem]{Proposition}

\newtheorem{coromain}{Corollary}

\theoremstyle{definition}

\theoremstyle{remark}

\newtheorem{claim}[theorem]{Claim}
\newtheorem{remark}[theorem]{Remark}

\def\Bd{\mathop{\mathrm{Bd}}}

\def\Cl{\mathop{\mathrm{Cl}}}

\def\menos{\backslash}

\def\R{\mathbb{R}}

\def\N{\mathbb{N}}
\def\S{\mathbb{S}}
\def\R{\mathbb{R}}

\renewcommand{\ell}{l}

\def\intervalog{[0,l]}
\def\intervalogab{(0,l)}

\def\cunon{c_1(n)}
\def\cdosn{c_2(n)}
\def\ctresn{c_3(n)}
\def\ccuatron{c_4(n)}
\def\ccincon{c_5(n)}
\def\cseisn{c_6(n)}
\def\csieten{c_7(n)}
\def\cochon{c_8(n)}
\def\cnueven{c_9(n)}
\def\cdiezn{c_{10}(n)}
\def\cjotan{c_j(n)}
\def\celen{c_l(n)}

\def\rn{r_n}
\def\rnMuno{r_{n+1}}

\def\pn{p_n}
\def\pnMuno{p_{n+1}}

\def\c1nuno{c_1(n+1)}

\def\cunonmuno{c_1(n-1)}
\def\cdosnmuno{c_2(n-1)}
\def\ctresnmuno{c_3(n-1)}
\def\ccuatronmuno{c_4(n-1)}
\def\ccinconmuno{c_5(n-1)}
\def\cseisnmuno{c_6(n-1)}
\def\csietenmuno{c_7(n-1)}
\def\cochonmuno{c_8(n-1)}
\def\cnuevenmuno{c_9(n-1)}
\def\cdieznmuno{c_{10}(n-1)}

\def\ctresnmuno{c_3(n-1)}

\def\cunonMuno{c_1(n+1)}
\def\cdosnMuno{c_2(n+1)}
\def\ctresnMuno{c_3(n+1)}
\def\ccuatronMuno{c_4(n+1)}
\def\ccinconMuno{c_5(n+1)}
\def\cseisnMuno{c_6(n+1)}

\def\defcdosnMuno{2c_1(n)+(2\rnMuno +2)c_2(n)+(\rnMuno +1)c_6(n)+(\rnMuno +2)c_7(n)+c_9(n)+(4\rnMuno +6)c_{10}(n)}
\def\defctresnMuno{2c_1(n)+2c_2(n)+2c_3(n)+c_5(n)+2c_6(n)+2c_7(n)+c_9(n)+10c_{10}(n)}
\def\defccuatronMuno{2c_1(n)+2c_2(n)+2c_3(n)+(\rnMuno +2)c_4(n)+2c_5(n)+2c_6(n)+2c_7(n)+c_9(n)+(\rnMuno +13)c_{10}(n)}
\def\defccinconMuno{2c_1(n)+2c_2(n)+2c_3(n)+(\rnMuno +1)c_4(n)+2c_5(n)+2c_6(n)+2c_7(n)+c_9(n)+(\rnMuno +12)c_{10}(n)}
\def\defcseisnMuno{2c_1(n)+2c_2(n)+c_3(n)+2c_6(n)+2c_7(n)+c_9(n)+8c_{10}(n)}

\newcommand{\ssstresn}{\sss_3(n)}
\newcommand{\bbbtresn}{\bbb_3(n)}

\newcommand{\sssseisn}{\sss_6(n)}
\newcommand{\bbbseisn}{\bbb_6(n)}

\newcommand{\bbbin}{\bbb_i(n)}
\newcommand{\sssin}{\sss_i(n)}

\newcommand{\cinMuno}{c_i(n+1)}

\title{Minimal non uniquely ergodic flipped \iets}

\author{Antonio Linero Bas}

\address{Departamento de Matem\'aticas\\
Universidad de Murcia (Campus de Espinardo)\\
30100-Espinardo-Murcia (Spain)
} 

\email{lineroba@um.es}

\author{
Gabriel Soler L\'{o}pez
}

\address{
Departamento de Matem\'atica Aplicada y Estad\'\i{}stica\\
Paseo Alfonso XIII, 52\\
30203-Cartagena (Spain)}

\email{gabriel.soler@upct.es}

\renewcommand{\arccos}{ {\mathrm{arccos}}\;}

\begin{document}

\begin{abstract}

In this paper we prove the existence of minimal non uniquely ergodic flipped \iets. In particular, we build explicitly minimal non uniquely ergodic $(10,k)$-\iets\, for any $1\leq k \leq 10$. This
answers an open question posed in  \cite[Remark 1]{gutierrez4b}. As a consequence, we also derive the existence of transitive non uniquely ergodic $(n,k)$-\iets, for any $n\geq 10$ and $1\leq k\leq n$ if $n$ is even, and $1\leq k\leq n-1$ if $n$ is odd. 

%
%
%

\end{abstract}



\maketitle

%

\section{Introduction}

Interval exchange transformations, for short \iets, have generated a continuous interest from the first work of Keane \cite{keane} and have given a huge amount of literature. Its study has two remarkable branches: oriented \iets\ and 
flipped \iets\ with significant differences in the branches. A lot of efforts have been made to develop the theory of \iets\ without flips, see \cite{viana} for an exhaustive review of the orientable case. 
However, the flipped case has advanced more slowly and  it still remains questions to clarify which are known for the oriented case from the seventies. One of these questions is to prove the existence of 
non uniquely ergodic minimal \iets, see \cite[Remark 1]{gutierrez4b}. This paper is devoted to close this gap. In the oriented case the existence of non uniquely ergodic minimal \iets\ was first clarified in \cite{keynesnewton}.

For the sake of completeness we recall some definitions from \cite{linerosolerergodic}. 
Let $n\in\N=\{1,2,3,\dots\}$ be a natural and $l$ a real positive number, an  
\emph{$n$-\iet} is  an injective map $T:D\subset\intervalogab\to \intervalogab$
such that:
\begin{description}
\item[(i)] $D$ is the union of $n$ pairwise disjoint open intervals, $D=\bigcup_{i=1}^n I_i$,
moreover $I_i=(a_i,a_{i+1})$, $0=a_1<a_2<a_3<\dots<a_{n+1}=l$;
\item[(ii)] $T|_{I_i}$ is an affine  map of constant slope equals to $1$ or $-1$, $i=1,2,\dots, n$.
\end{description}

When the slope of $T$ is negative in the interval set $\mathcal{F}=\{I_{f_1},I_{f_2},\ldots, I_{f_k}\}$, $k\leq n$, we say that   
$T$ is \emph{an interval exchange transformation of
$n$ intervals with $k$ flips}  or simply \emph{an (n,k)-\iet}; otherwise we
 say that $T$ is \emph{an interval exchange transformation of
$n$ intervals without flips} or simply \emph{an oriented interval exchange
transformation of $n$ intervals}. We will say that $T$ is a \emph{proper} $(n,k)$-\iet\ if the points $a_i$, $2\leq i\leq n$, are not fake discontinuities.

%
%
%
%
%
%

 The orbit of  $x\in \intervalogab$, generated by   $T$, is  the set
$$\O_T(x)=\{T^m(x): m \textrm{ is an integer and } T^m(x) \textrm{ makes sense} \},$$
where $T^0=\mathrm{Id}$ and  $T^m=T\circ T^{m-1}$ for any integer $m$.
Moreover $\O_T(0)=\{0\}\cup\O_T(\lim_{x\to0^+}T(x))$
and 
$\O_T(l)=\{l\}\cup\O_T(\lim_{x\to l^-}T(x))$. 
$T$ is said to be \emph{minimal} if $\O_T(x)$ is dense in $\intervalog$ 
for any  $x\in\intervalog$ while $T$ is \emph{transitive} if it has a dense orbit in   $\intervalog$
(this notion of minimality is equivalent to say that $T$ is transitive and it does not have finite orbits, see \cite[Remark~1]{linerosolerergodic}).

A finite measure  $\mu$   on $\intervalog$  is said to be  \emph{invariant for $T$} if for any measurable set $A\subset\intervalog$  
$\mu(T^{-1}(A))=\mu(A)$. An invariant measure $\mu$ for  $T$ is \emph{ergodic} if for any set    
$E\subset\intervalog$ verifying {$T^{-1}(E)=E$} then    $\mu(E)=0$ or
{$\mu(E)=1.$} Observe that the standard Lebesgue measure on $\intervalog$, denoted by $\mu_L$, is invariant for any interval exchange transformation $T$ and  any multiple of $\mu_L$ is also invariant.
$T$ is said to be \emph{uniquely ergodic} if it does not admit another invariant probability measure than the normalized Lebesgue one.
We stress that, for \iets, the unique ergodicity implies the ergodicity 
with respect to Lebesgue measure, cf. \cite[Section II.6, Th. 6.1]{MagneLibro}.

The objective of this paper is to prove the following theorem.

\begin{mteo}
 There exist minimal non uniquely ergodic flipped \iets.  In particular we build proper minimal non uniquely ergodic  $(10,k)-$\iets\ for any $1\leq k\leq 10$.
\end{mteo}

By using a construction proposed in \cite{gutierrez4b} we will be able to deduce the existence of transitive non uniquely ergodic \iets.
\begin{coromain}\label{C:transitivo}
There exist transitive non uniquely ergodic proper $(n,k)$-\iets\ for all $n\geq 10$ and $1\leq k\leq n$ if $n$ is even and $1\leq k<n$ whenever $n$ is odd. 

Moreover, it is also possible to build transitive non uniquely ergodic oriented proper $n$-\iets\ for any $n\geq 4$ and transitive non uniquely ergodic proper $(n,2)$-\iets\  $(6\leq n\leq 9)$ and  
$(n,4)$-\iets\  $(8\leq n\leq 9)$.
\end{coromain}

The paper is organized as follows. In Section~\ref{RauzyVeechInduction} we recall the basic notions about \iets, in particular the generalized Rauzy maps and the Rauzy graph of an \iet\, $T$, and we stress their relationship with the cone $\mathcal{M}(T)$ of invariant measures associated to $T$, see Theorem~\ref{T-cone}. In Section~\ref{S:path} we present our candidate to minimal non uniquely ergodic \iet. It is worth mentioning that our inspiration was the minimal \iet\, constructed in 
\cite{linerosolerergodic}, and the idea of describing a non-periodic loop was suggested by the reading of \cite[S. IV]{YoccozCurso} and \cite{chaikamasur}. Next, in Section~\ref{S:matricespath} we deeply analyze  the relationships between the columns of successive matrices of the graph of $T$, and the main result of this part is Theorem~\ref{T:limiteplano}, which establishes the existence of exactly two limit directions for these columns. Section~\ref{S:proofmaintheorem} is devoted to prove our Main Theorem and Corollary~\ref{C:transitivo}. Finally, we present some other interesting questions for future studies.

\def\intervalog{[0,l]}
\def\intervalogab{(0,l)}

\def\codes{\mathcal{C}_n}

\section{Coordinates in the set of \iets, Rauzy induction and invariant measures}\label{RauzyVeechInduction}

\def\codes{\mathcal{C}_n}

An easy way to work with \iets\ is   introducing coordinates, we now follow  \cite{linerosolerergodic}. 
To do that, 
let $n\in \N$, then it is known the existence of a  natural
injection between the set of  $n$-\iets\  and $\codes=\Lambda^n\times
S_n^\sigma$, where $\R_+=(0,\infty)$, $\Lambda^n$ is the cone $\R_+^n$ and $S_n^\sigma$ is the 
set of \emph{signed permutations}. A signed permutation is an injective map $\pi:N_n=\{1,2,\ldots,n\}\to
N_n^\sigma=\{-n,-(n-1),\ldots,-1,1,2,\ldots,n\}$ such that $|\pi|:N_n\to
N_n$ is bijective, that is, a \emph{standard permutation}; a \emph{non
standard permutation} will be a signed permutation $\pi$ such that $\pi(i)<0$ for
some $i$. As in the case of standard permutations,
$\pi$ will be represented by the vector $(\pi(1),\pi(2),\ldots,\pi(n))\in
(N_n^\sigma)^n.$ Let $T$ be an $n$-\iet\ like in the preceding paragraph,
then its associated coordinates in $\codes$ are $(\lambda,\pi)$ defined
by:
\begin{itemize}
\item \label{notacion}$\lambda_i=a_{i+1}-a_i$ for all $i\in N_n$.
\item $\pi(i)$ is positive (resp. negative) if $T|_{I_i}$ has
slope $1$ (resp. $-1$). Moreover $|\pi(i)|$ is the position of the
interval $T(I_i)$ in the set $\{T(I_i)\}_{i=1}^n$ taking into account the
usual order in $\R$.
\end{itemize}
Conversely, given  a pair $(\lambda,\pi)\in\codes$ we can associate to it  a unique
$n$-\iet, $T:D\subset \intervalog\to\intervalog$, where:
\begin{itemize}
\item $l=|\lambda|:=\sum_{i=1}^n\lambda_i$;
\item  $I_1=(0,\lambda_1)$;
\item $I_i=(\sum_{j=1}^{i-1} \lambda_j,\sum_{j=1}^{i} \lambda_j)$ for any
$1<i\leq n$;
\item
$T|_{I_i}(x)=\left(\sum_{j=1}^{|\pi|(i)-\frac{\sigma(\pi(i))+1}{2}}\lambda_{|\pi|^{-1}(j)}\right)
+\sigma(\pi(i))\left[x -\left(\sum_{j=1}^{i-1}\lambda_{j}\right)\right],$ for any $1\leq i\leq n$, where $\sigma(z)$ denotes the \emph{sign} of $z\in\mathbb{R}\setminus\{0\},$ namely, $\sigma(z)=\frac{z}{|z|}.$
\end{itemize}
These coordinates allow us to make the identification
$T=(\lambda,\pi)$. 

\def\O{\mathcal{O}}

\def\d{\delta}

In \cite{linerosolerergodic} the {authors constructed} minimal uniquely ergodic 
interval exchange transformations with
flips, generalizing the results in  \cite{gutierrez4b}.  However, in this {last} paper,
it is left as an open problem to prove the existence of minimal non uniquely ergodic minimal flipped \iets.
We construct, in this work, an example of this type of \iets.

\def\TT{\overline{T}}

A permutation $\pi:N_n\to
N_n^\sigma$ is said to be irreducible if $|\pi(\{1,2,\ldots t\})|\not=\{1,2,\ldots,t\}$ for any $1\leq t<  n$. The set of \corregido{\emph{irreducible permutations}}{} is denoted by
$S_n^{\sigma,*}$. \corregido{We will write $S_n^{\sigma,+}$ to denote the set of permutations, $\pi\in S_n^\sigma$, satisfying $|\pi|(n)\not=n$. 
Observe that $S_n^{\sigma,*}\subset S_n^{\sigma,+}\subset S_n^{\sigma}$.}{} It is easily seen that if $(\lambda,\pi)$ is a minimal
$n$-\iet\ (not necessarily oriented) then $\pi$ is irreducible.

\def\intervalogp{[0,l']}
\def\irreducibleset{S_n^{\sigma,*}}
\def\permuset{S_n^{\sigma}}
\def\permusetplus{S_n^{\sigma,+}}

Roughly speaking, the \emph{generalized Rauzy induction} is an
operator in the set of \iets\  which sends any $T:D\subset \intervalog\to \intervalog$
to its first return map on some subinterval
$\intervalogp\subsetneq\intervalog.$ We pass  to give a
formalization of this operator, by means of the maps $a$ and $b$ defined on $S_n^\sigma$. In the 
final part of the section we investigate the relationship between the Rauzy induction and  the existence 
of minimal  \iets\ with flips {(see Theorem~\ref{T-cone}).}

Let $x\in\R\menos\{0\}$. Recall that the sign of $x$ is denoted by $\sigma(x).$
The \emph{generalized Rauzy maps} were introduced by Nogueira in \cite{nogueira2} 
(cf. also \cite{rauzy}) and are   $
a,b:\permusetplus\rightarrow\permuset,$
where $a(\pi)$ and $b(\pi)$ are the permutations defined  by:
\begin{equation}\label{def:a}
a(\pi)(i)=\left\{
\begin{array}{ll}
\pi(i)                          &\textrm{ if }|\pi(i)|\leq |\pi(n)| - \frac{1-\sigma(\pi(n))}{2},\\
\sigma(\pi(n))\sigma(\pi(i))(|\pi(n)|+\frac{1+\sigma(\pi(n))}{2})        &\textrm{ if }|\pi(i)|=n,\\
\sigma(\pi(i))(|\pi(i)|+1)        &\textrm{ otherwise,}\\
\end{array}\right.
\end{equation}
and 
\begin{equation}\label{def:b}
b(\pi)(i)=\left\{
\begin{array}{ll}
\pi(i)                          &\textrm{ if } i \leq |\pi|^{-1}(n) + 
\frac{\sigma(\pi(|\pi|^{-1}(n)))-1}{2},\\
\sigma(\pi(|\pi|^{-1}(n)))\pi(n)        &\textrm{ if } i = |\pi|^{-1}(n) + 
\frac{\sigma(\pi(|\pi|^{-1}(n)))-1}{2} +1, \\
\pi(i-1)        &\textrm{ otherwise.}\\
\end{array}\right.
\end{equation}

Together with these maps, we also define the \emph{generalized
Rauzy matrices} associated to a permutation $\pi\in\permusetplus$, $M_a(\pi)$ and
$M_b(\pi)$. Given  $1\leq i,j\leq n$, $E_{i,j}$ denotes the $n\times n$
matrix having zeros in all the positions except \corregido{for}{} the position $(i,j)$
which is equal to 1, and $I_n$ denotes the $n\times n$ identity matrix. The
definitions of $M_a(\pi)$ and $M_b(\pi)$ are:
\begin{eqnarray}
M_a(\pi)&=&I_n+E_{n,|\pi|^{-1}(n)};\nonumber\\
M_b(\pi)&=&\left(\sum\limits_{i=1}^{|\pi|^{-1}(n)}E_{i,i}\right)
+ E_{n, |\pi|^{-1}(n)+ \frac{1+\sigma(\pi(|\pi|^{-1}(n)))}{2}} 
+ \left(\sum\limits_{i=|\pi|^{-1}(n)}^{n-1}E_{i,i+1}\right).
\label{E:matrizAB}
\end{eqnarray}
%
%
%
%
%
%
%

Positive matrices will play an important role in our study on minimality of \iets. 
A \emph{non-negative} matrix $A\in M_{n\times n}(\R)$, i.e. $a_{i,j}\geq 0$ for
any $i,j\in\{1,2,\ldots,n\}$, is said to be \emph{positive} if the
previous inequalities are strict. 
In the following, the diagonal, a row or a column of a  matrix is said to be positive if all the entries in the corresponding diagonal, row or column are positive.

\corregido{We are now}{} ready to present formally the \emph{generalized Rauzy 
operator} $R$. Let 
$$\mathcal{D}=\{(\lambda,\pi)\in\Lambda^n\times
S_n^\sigma: \lambda_{n}\not=\lambda_{|\pi|^{-1}(n)}\},$$ then
$$
\begin{array}{rcl}
R:\mathcal{D}\subset \Lambda^n\times S_n^\sigma&\longrightarrow&
\Lambda^n\times S_n^\sigma\\
T=(\lambda,\pi)&\to&T'=(\lambda',\pi')
\end{array}$$
is defined by:

$$T'=(\lambda',\pi')=\left\{
\begin{array}{ll}
\left(M_a(\pi)^{-1}\lambda,a(\pi)\right)&\textrm{if }
\lambda_{|\pi|^{-1}(n)}<\lambda_n,\\
\left(M_b(\pi)^{-1}\lambda,b(\pi)\right)&\textrm{if }
\lambda_{|\pi|^{-1}(n)}>\lambda_n.\\
\end{array}\right.$$

If $T'$ is obtained from $T$ by means of the operator $a,$ $T$ is said to be of \emph{type} $a$, otherwise $T$ is of \emph{type} $b$. In any case, $T'$  is the \emph{Poincar\'e first return map} induced by $T$ on $[0,l'],$ with $l'=l-\min\{\lambda_n,\lambda_{\antin}\}$, see 
\cite[Proposition 5]{mioangosto}.  

The operators $a$ and $b$ induce in the set $S_n^{\sigma,*}$ a directed graph
structure whose vertices are all the points from $S_n^{\sigma,*}$ and the
directed edges are arrows labelled by $a$ and $b$. Given $\pi,\pi'\in
S_n^{\sigma,*},$ there exits an arrow labelled by $a$ (resp. $b$) from $\pi$ to
$\pi'$ if and only if $a(\pi)=\pi'$ (resp. $b(\pi)=\pi'$). Any connected
subgraph of this graph, $\mathcal{G}_n$, is called a \emph{Rauzy class} 
\corregido{(the Rauzy classes for standard permutations were studied in \cite{kontsevichzorich}).}{}
We remark that we only
take into account irreducible permutations because they are the only ones for which the associated
\iets\ can be minimal. Moreover, it is worth noticing that if $\pi$ is an irreducible standard permutation then 
$a(\pi)$ and $b(\pi)$ are irreducible, while it is not always the case for non standard irreducible permutations, 
observe for instance that $a(-4,3,2,-1)=(1,4,3,-2)$.

A \emph{vector of operators} is an element of $\{a,b\}^L$, where $L\in\N$
or $L=\infty$ (when  $L=\infty$, $\{a,b\}^L=\{a,b\}^{\N}$). An easy way of constructing Rauzy subgraphs
is to take a vertex $\pi\in S_n^{\sigma,*}$ and to construct recursively other
vertices by applying a vector of operators. The \emph{Rauzy subgraph}, $\mathcal{G}^{\pi_1,v}$,
associated to $\pi_1\in S_n^{\sigma,*}$ and $v=(v_1,v_2,\dots)\in \{a,b\}^L$
 is the graph of vertices $\{\pi_i\}_{i=1}^L$
satisfying $v_i(\pi_i)=\pi_{i+1}$, $1\leq i < L$, the edges of this graph being arrows
labelled by $v_i$ from $\pi_i$ to $\pi_{i+1}$. Observe that any $n$-\iet,
$T=(\lambda,\pi)\in\mathcal{D}$, defines a Rauzy subgraph in a natural way, the one
associated to $\pi$ and the vector of operators $v$ defined by the Rauzy 
induction, that is, $v_i$ is $a$ (resp $b$) if $R^{i-1}(T)$ is of type $a$ (resp. $b$), we denote this subgraph by $\mathcal{G}^T$. We will say that $T$
is \emph{infinitely inducible} if $v$ has infinite dimension, i.e.
$v\in\{a,b\}^{\N}$. For a finite vector of operators, $v\in\{a,b\}^k${, $v=(v_1,\ldots,v_k)$}, $k\in\N$, and a vertex
$\pi\in S_n^{\sigma,*}$, $v(\pi)$ denotes the vertex obtained after applying sequentially, from the left, the operators in $v$, also
$M_v(\pi):=
M_{v_1}(\pi)
M_{v_2}(v_1(\pi))
M_{v_3}(v_2(v_1(\pi)))
\dots
M_{v_k}(v_{k-1}(v_{k-2}(\dots v_1(\pi))))
$.

We are now in a position to establish our first result concerning the relationship between 
Rauzy subgraphs and $\mathcal{M}(T).$

\newcommand{\matrizpaso}[1]{M_{v_{#1}}(\pi_{#1})}

\begin{theorem}\label{T-cone}
Let $T=(\lambda^1,\pi_1)$ be an $n$-\iet\ such that $\pi_1\in S_n^{\sigma,*}$,
$T$ is infinitely inducible  and 
$R^i(\lambda^1,\pi_1)=(\lambda^{i+1},\pi_{i+1})$ {for any $i\geq 1$;  let $\mathcal{G}^T$ be the  Rauzy subgraph of
$T$ associated to $\pi_1$ and
$v\in\{a,b\}^\N$.}  Put 
\begin{equation}
\mathcal{C}(\mathcal{G}^T):=\bigcap_{i=1}^\infty M_{v_1}(\pi_1)\cdot
M_{v_2}(\pi_2)\cdot \ldots \cdot M_{v_i}(\pi_i) \Lambda^n,
\label{E:cone}
\end{equation}
and let $\mathcal{M}(T)$ be  the cone of invariant measures associated to $T$.
Assume also that, for any $i\in\N$, $\pi_i$ is irreducible. Then:
\begin{enumerate}
\item $\lambda^1\in\mathcal{C}(\mathcal{G}^T)$;
\item if $\gamma\in \mathcal{C}(\mathcal{G}^T)$ and $S=(\gamma,\pi_1)$, \quad
$\mathcal{G}^S=\mathcal{G}^T$.
\item $\mathcal{C}(\mathcal{G}^T)$ and  $\mathcal{M}(T)$ are linearly isomorphic (thus,  if $\mathcal{C}(\mathcal{G}^T)$ is a half line, $T$ is uniquely ergodic).
\item $R^j(T)$ is minimal for any $j\in\N\cup\{0\}$.
\end{enumerate}
\end{theorem}

\begin{proof}
 See \cite[Th. 20]{linerosolerergodic} for the proof of items (1) and (2). The proof of third item can be followed for oriented \iets\ in  \cite[Section~28]{viana},  we stress that 
 the proof
 also applies in the non-oriented case because the non-orientability only plays an essential role in Lemma~28.2. The analogous of this lemma in the flipped case is stated in 
 \cite[Th. 22]{linerosolerergodic}. Item (4) is proved in \cite{mioangosto}.
\end{proof}

\newcommand{\conografo}[1]{\mathcal{C}(\mathcal{G}^{#1})}

\newcommand{\ensayo}[1]{#1
}

Next result gives a method for constructing minimal \iets\ by means of Rauzy
graphs. The proof will be made in Section~\ref{S:proofmaintheorem}.

\begin{theorem}\label{conver}\label{T:xxlnoperiodico} Let
$\mathcal{G}^{\pi_1,v}$ be a  the Rauzy-subgraph associated  to
$\pi_1\in S_n^{\sigma,*}$ and $v\in\{a,b\}^\N$. Assume the existence of a   sequence  $(n_k)_k${,} $n_k\in\N$, satisfying 
$n_1=1$, $n_{k+1}>n_k$  and  {such that}
$M_{v_{n_k}}(\pi_{n_k})\cdot
M_{v_{n_k+1}}(\pi_{n_k+1})\cdot \ldots \cdot M_{v_{n_{k+1}-1}}(\pi_{v_{n_{k+1}-1}})$ {is positive} for any $k\in\N$. {If} $\mathcal{C}(\mathcal{G}^{\pi_1,v}):=\bigcap_{i=1}^\infty M_{v_1}(\pi_1)\cdot
M_{v_2}(\pi_2)\cdot \ldots \cdot M_{v_i}(\pi_i) \Lambda^n$, then: 
\begin{enumerate}
\item  $\mathcal{C}(\mathcal{G}^{\pi_1,v})$ is nonempty. 
 \item There exists $\lambda^1\in\mathcal{C}(\mathcal{G}^{\pi,v})$ such that the associated graph to $T$ is
$\mathcal{G}^{\pi_1,v}$.
\item $R^j(T)$ is minimal for any $j\in\N\cup\{0\}$.
\item $T$ is uniquely ergodic if and only if $R^j(T)$ is  uniquely ergodic for any $j\in\N\cup\{0\}$.
\end{enumerate}
\end{theorem}

Taking into account Theorems~\ref{T-cone}, \ref{conver} 
and \cite[Th. 25]{linerosolerergodic} we must find $v\in\{a,b\}^\N$ generating non periodic Rauzy-subgraphs to obtain non uniquely ergodic minimal \iets. In the next section we construct this subgraph.

\def\ctoB{\bigcap_{j\in\N}(M_{\pi_1,\mathbf{v}}^\mathcal{G})^{sj}
\Lambda^n}
\def\ctoA{\bigcap_{i=1}^\infty M_{v_1}(\pi_1)\cdot M_{v_2}(\pi_2)\cdot\ldots\cdot M_{v_i}(\pi_i)
\Lambda^n}
\def\MatrizM{M^\mathcal{G}_{\pi_1,v}}
\def\MatrizMextendida{M_{v_1}(\pi_1)\cdot M_{v_2}(\pi_2)\cdots M_{v_p}(\pi_p)}

\section{The path we follow}\label{S:path}

In this section, we build our candidate to minimal non uniquely ergodic flipped IET map. Previously, in Subsection~\ref{Subs:core} we present the \textit{core path}, which is the same we used in \cite{linerosolerergodic} for the construction of minimal IETs with flips. Since we know that the   associated graph cannot be periodic {(see \cite[Th. 25]{linerosolerergodic})}, we need to modify this core path in the following way: we detect its fixed vertices by the Rauzy operators $a$ or $b$, and then in each periodic tour of the core path we apply arbitrarily many times the operator $a$ or $b$ to some of these fixed vertices, having the precaution of increasing in each step the number, as times as necessary, of such applications. In this manner, we avoid to have a periodic graph and we construct our suitable path in Subsection~\ref{Subs:path}. Then, once we have presented our candidate to minimal non-uniquely ergodic map $T$, the rest of sections are devoted to stress the relationships between the columns of the matrices appearing in the cone $\mathcal{C}(\mathcal{G}^T),$ in order to prove that this cone is two-dimensional, and therefore $T$ is {non-uniquely} ergodic.

\subsection{Core path}\label{Subs:core}

In view of {Theorems~\ref{T-cone} and \ref{T:xxlnoperiodico},} we need to find a vector $v\in\{a,b\}^{\N}$ and a permutation $\pi\in S_n^{\sigma,*}$ such that 
$\mathcal{C}(\mathcal{G}^{\pi,v})$ has dimension bigger {than} or equal to 2. 
Our first step will be to choose the initial permutation and the core path or the path we will take as the basis for doing an appropriate repetition in the form of loops of length bigger and bigger.
We will focus on $10-\iets$, we will take
$\pi_0=(-3,-4,-5,-6,-7,-8,-9,10,1,-2)$ and the vectors

\begin{eqnarray}
v^1&=&(a,a,a,a,a,a,a,b,a,b,b,a,b,a,b,a,b)\in\{a,b\}^{17},\\
v^2(r)&=&(\underbrace{b,b,\dots,b,b}_{r})\in\{a,b\}^{r}, r\geq 0,\\
v^3&=&(
a,
\underbrace{b,b}_{2},
\underbrace{a,a,a}_{3},
\underbrace{b,b,b,b}_{4},
\underbrace{a,a,a,a,a}_{5},
\underbrace{b,b,b,b,b}_{5})\in\{a,b\}^{20},\\
v^4(s)&=&(\underbrace{a,a,\dots,a,a}_{s})\in\{a,b\}^{s}, s\geq 0,\\
v^5&=&(b,
\underbrace{a,a,a,a,a,a,a}_{7},b,a,b)
\in\{a,b\}^{11},\\
v(r,s)&=&v^1*v^2(r)*v^3*v^4(s)*v^5\in\{a,b\}^p,\quad
p=48+r+s.\label{eqV}
\end{eqnarray}

This path is a generalization of the  employed in \cite{linerosolerergodic} for the construction of general non-orientable minimal $(n,k)$-\iets, $1\leq k\leq n$. 
Notice that $v(0,0)=v^1*v^2(0)*v^3*v^4(0)*v^5=v^1*v^3*v^5$, in fact 
$$v(0,0)=aaaaaaababbababababbaaabbbbaaaaabbbbbbaaaaaaabab=a^7bab^2(ab)^3*ab^2a^3b^4a^5b^5*ba^7bab,$$
where $*$ is meant the concatenation of vectors.

We begin with the vertex $\pi_0$ and we apply to it sequentially the operators of  $v(0,0)$ beginning from the left. Let  $w^0:=v(0,0)$, and let $\mathcal{G}^{\pi_0,w^0}$ be  the graph of vertices $\{\pi_i\}_{i=0}^{48}$. Then: 

\begin{lemma}\label{L:permutations}
 The permutations in the graph $\mathcal{G}^{\pi_0,w^0}$ are:
 
 \noindent
 \begin{tabular}{lp{.3cm}l}
 $\pi_{0}=(-3,-4,-5,-6,-7,-8,-9,10,1,-2)$,&&
 $\pi_{1}=a(\pi_{0})=(-4,-5,-6,-7,-8,-9,-10,-2,1,-3)$,\\
 $\pi_{2}=a(\pi_{1})=(-5,-6,-7,-8,-9,-10,3,-2,1,-4)$,&&
 $\pi_{3}=a(\pi_{2})=(-6,-7,-8,-9,-10,4,3,-2,1,-5)$,\\
 $\pi_{4}=a(\pi_{3})=(-7,-8,-9,-10,5,4,3,-2,1,-6)$,&&
 $\pi_{5}=a(\pi_{4})=(-8,-9,-10,6,5,4,3,-2,1,-7)$,\\
 $\pi_{6}=a(\pi_{5})=(-9,-10,7,6,5,4,3,-2,1,-8)$,&&
 $\pi_{7}=a(\pi_{6})=(-10,8,7,6,5,4,3,-2,1,-9)$,\\
 $\pi_{8}=b(\pi_{7})=(9,-10,8,7,6,5,4,3,-2,1)$,&&
 $\pi_{9}=a(\pi_{8})=(10,-2,9,8,7,6,5,4,-3,1)$,\\
 $\pi_{10}=b(\pi_{9})=(10,1,-2,9,8,7,6,5,4,-3)$,&&
 $\pi_{11}=b(\pi_{10})=(10,-3,1,-2,9,8,7,6,5,4)$,\\
 $\pi_{12}=a(\pi_{11})=(5,-3,1,-2,10,9,8,7,6,4)$,&&
 $\pi_{13}=b(\pi_{12})=(5,-3,1,-2,10,4,9,8,7,6)$,\\
 $\pi_{14}=a(\pi_{13})=(5,-3,1,-2,7,4,10,9,8,6)$,&&
 $\pi_{15}=b(\pi_{14})=(5,-3,1,-2,7,4,10,6,9,8)$,\\
 $\pi_{16}=a(\pi_{15})=(5,-3,1,-2,7,4,9,6,10,8)$,&&
 $\pi_{17}=b(\pi_{16})=(5,-3,1,-2,7,4,9,6,10,8)$,\\
 $\pi_{18}=a(\pi_{17})=(5,-3,1,-2,7,4,10,6,9,8)$,&&
 $\pi_{19}=b(\pi_{18})=(5,-3,1,-2,7,4,10,8,6,9)$,\\
 $\pi_{20}=b(\pi_{19})=(5,-3,1,-2,7,4,10,9,8,6)$,&&
 $\pi_{21}=a(\pi_{20})=(5,-3,1,-2,8,4,7,10,9,6)$,\\
 $\pi_{22}=a(\pi_{21})=(5,-3,1,-2,9,4,8,7,10,6)$,&&
 $\pi_{23}=a(\pi_{22})=(5,-3,1,-2,10,4,9,8,7,6)$,\\
 $\pi_{24}=b(\pi_{23})=(5,-3,1,-2,10,6,4,9,8,7)$,&&
 $\pi_{25}=b(\pi_{24})=(5,-3,1,-2,10,7,6,4,9,8)$,\\
 $\pi_{26}=b(\pi_{25})=(5,-3,1,-2,10,8,7,6,4,9)$,&&
 $\pi_{27}=b(\pi_{26})=(5,-3,1,-2,10,9,8,7,6,4)$,\\
 $\pi_{28}=a(\pi_{27})=(6,-3,1,-2,5,10,9,8,7,4)$,&&
 $\pi_{29}=a(\pi_{28})=(7,-3,1,-2,6,5,10,9,8,4)$,\\
 $\pi_{30}=a(\pi_{29})=(8,-3,1,-2,7,6,5,10,9,4)$,&&
 $\pi_{31}=a(\pi_{30})=(9,-3,1,-2,8,7,6,5,10,4)$,\\
 $\pi_{32}=a(\pi_{31})=(10,-3,1,-2,9,8,7,6,5,4)$,&&
 $\pi_{33}=b(\pi_{32})=(10,4,-3,1,-2,9,8,7,6,5)$,\\
 $\pi_{34}=b(\pi_{33})=(10,5,4,-3,1,-2,9,8,7,6)$,&&
 $\pi_{35}=b(\pi_{34})=(10,6,5,4,-3,1,-2,9,8,7)$,\\
 $\pi_{36}=b(\pi_{35})=(10,7,6,5,4,-3,1,-2,9,8)$,&&
 $\pi_{37}=b(\pi_{36})=(10,8,7,6,5,4,-3,1,-2,9)$,\\
 $\pi_{38}=b(\pi_{37})=(10,9,8,7,6,5,4,-3,1,-2)$,&&
 $\pi_{39}=a(\pi_{38})=(-2,10,9,8,7,6,5,-4,1,-3)$,\\
 $\pi_{40}=a(\pi_{39})=(-2,-3,10,9,8,7,6,-5,1,-4)$,&&
 $\pi_{41}=a(\pi_{40})=(-2,-3,-4,10,9,8,7,-6,1,-5)$,\\
 $\pi_{42}=a(\pi_{41})=(-2,-3,-4,-5,10,9,8,-7,1,-6)$,&&
 $\pi_{43}=a(\pi_{42})=(-2,-3,-4,-5,-6,10,9,-8,1,-7)$,\\
 $\pi_{44}=a(\pi_{43})=(-2,-3,-4,-5,-6,-7,10,-9,1,-8)$,&&
 $\pi_{45}=a(\pi_{44})=(-2,-3,-4,-5,-6,-7,-8,-10,1,-9)$,\\
 $\pi_{46}=b(\pi_{45})=(-2,-3,-4,-5,-6,-7,-8,9,-10,1)$,&&
 $\pi_{47}=a(\pi_{46})=(-3,-4,-5,-6,-7,-8,-9,10,-2,1)$,\\
 $\pi_{48}=b(\pi_{47})=(-3,-4,-5,-6,-7,-8,-9,10,1,-2){=\pi_0}$,&&
 \end{tabular}

 \end{lemma}

 Realize that any permutation $\pi$ of type $\pi=(\pi(1),\ldots,\pi(9), 9)$ is fixed by the operator $a$, and on the other hand 
a permutation $\pi$ is fixed by the operator $b$ if   $\pi=(\pi(1),\ldots,\pi(8), 10,\pi(10))$. Consequently:

\begin{lemma}\label{L:permutations2} 
In the graph  $\mathcal{G}^{\pi_0,w^0}$ we find the following fixed vertices by the Rauzy operator $a$: 

$\pi_{19}=a(\pi_{19})= (5,-3,1,-2,7,4,10,8,6,9),$
 
$\pi_{26}=a(\pi_{26})= (5,-3,1,-2,10,8,7,6,4,9),$
 
$\pi_{37}=a(\pi_{37})=(10,8,7,6,5,4,-3,1,-2,9).$

The following vertices are fixed by the Rauzy operator $b$:

 $\pi_{16}=b(\pi_{16})=(5,-3,1,-2,7,4,9,6,10,8),$

 $\pi_{22}=b(\pi_{22})= (5,-3,1,-2,9,4,8,7,10,6),$
 
 $\pi_{31}=b(\pi_{31})= (9,-3,1,-2,8,7,6,5,10,4).$
\end{lemma}
Once we have defined the graph $\mathcal{G}^{\pi_0,w^0}$ corresponding to the vertices $\{\pi_i\}_{i=0}^{48}$, we are going to consider 
the vectors $v(r,s)=v^1*v^2(r)*v^3*v^4(s)*v^5$ having length equal to $p=48+r+s$ for any non-negative integers $r,s$. \linero{Realize that $v(r,s)$ is the core path  with $r$ applications of the operator $b$ to the permutation $\pi_{16}$ of $w^0$, and $s$ applications of the operator $a$ to $\pi_{37}$.  } Associate to it and 
the Rauzy process, we find the matrices for each one of the vectors whose concatenation originates $v(r,s)$: 
\begin{eqnarray*}
M_1&:=& M_{v_1^1}(\pi_0)\cdot M_{v_2^1}(\pi_1)\cdot \ldots \cdot M_{v_{17}^1}(\pi_{16}),\\
M_2(r)&:=&M_{v_{1}^2}(\pi_{17})\cdot M_{v_{2}^2}(\pi_{18})\cdot \ldots \cdot M_{v_{r}^2}(\pi_{17+r-1})=\left(M_b(\pi_{17})\right)^r,\\
M_3&:=& M_{v_{1}^3}(\pi_{17+r})\cdot M_{v_{2}^3}(\pi_{17+r+1}) \cdot \ldots \cdot M_{v_{20}^3}(\pi_{17+r+19}), \\
M_4(s)&:=& M_{v_{1}^4}(\pi_{17+r+20})\cdot M_{v_{2}^4}(\pi_{17+r+21}) \cdot \ldots \cdot M_{v_{s}^4}(\pi_{17+r+19+s})=\left(M_a(\pi_{37})\right)^s,\\
M_5&:=& M_{v_{1}^5}(\pi_{17+r+20+s})\cdot M_{v_{2}^5}(\pi_{17+r+21+s})\cdot \ldots \cdot M_{v_{11}^5}(\pi_{17+r+30+s})
\end{eqnarray*} 

With respect to $M_1$, a rather cumbersome calculation gives

{\begin{center}
$M_1=M_{v_1^1}(\pi_0)\cdot
M_{v_2^1}(\pi_1)\cdot \ldots \cdot M_{v_{17}^1}(\pi_{16})=
\begin{pmatrix}
1&1&1&1&0&0&0&0&0&0\\ 0&0&0&0&1&1&0&0&0&0\\ 0&0&0&0&0&0&
 1&1&0&0\\ 0&0&0&0&0&0&0&0&1&1\\ 0&0&0&0&0&0&1&0&0&1\\ 0&0&0&0&1&0
 &0&1&0&0\\ 1&0&0&0&0&1&0&0&0&0\\ 0&1&0&0&0&0&0&0&0&0\\ 0&0&1&1&0&
 0&0&0&0&0\\ 2&2&1&0&2&2&2&2&1&2\\ 
\end{pmatrix}.
$
 \end{center}
}

In relation with the value of $M_2(r)$, take into account that $\left(M_b(\pi_{17})\right)^r=(I_{10}+E_{9,10})^r=I_{10}+ rE_{9,10},$ therefore
$$M_2(r)=M_{v_{1}^2}(\pi_{17})\cdot
M_{v_{2}^2}(\pi_{18})\cdot \ldots \cdot M_{v_{r}^2}(\pi_{17+r-1})=
\begin{pmatrix} 
1&0&0&0&0&0&0&0&0&0\\ 0&1&0&0&0&0&0&0&0&0\\ 0&0&1&0&0&0& 0&0&0&0\\ 0&0&0&1&0&0&0&0&0&0\\ 0&0&0&0&1&0&0&0&0&0\\ 0&0&0&0&0&1 &0&0&0&0\\ 0&0&0&0&0&0&1&0&0&0\\ 0&0&0&0&0&0&0&1&0&0\\ 0&0&0&0&0& 0&0&0&1&r\\ 0&0&0&0&0&0&0&0&0&1\\
\end{pmatrix}.
$$
Concerning $M_3,$ a direct computation gives 
$$M_3=M_{v_{1}^3}(\pi_{17+r})\cdot
M_{v_{2}^3}(\pi_{17+r+1})
\cdot \ldots \cdot 
M_{v_{20}^3}(\pi_{17+r+19})=
\begin{pmatrix} 
1&1&1&1&1&1&0&0&0&0\\ 0&0&0&0&0&0&1&0&0&0\cr 0&0&0&0&0&0&
 0&1&0&0\cr 0&0&0&0&0&0&0&0&1&0\cr 0&1&1&1&1&0&0&0&0&1\cr 0&1&1&1&1&1
 &0&0&0&1\cr 0&1&1&1&0&0&0&0&0&0\cr 0&1&1&1&1&0&0&0&0&0\cr 0&0&1&0&0&
 0&0&0&0&0\cr 0&0&1&1&0&0&0&0&0&0\cr 
\end{pmatrix}.
$$
To compute $M_4(s)$ observe that $M_{v_{j}^4}(\pi_{37+r-j})=M_b(\pi_{37})=I_{10}+E_{10,1}$ for $j\in\{0,1,\ldots,s-1\}$. Then it is easily seen 
that  $\left(M_b(\pi_{37})\right)^s=I_{10}+sE_{10,1}$, thus 
%

$$M_4(s)=M_{v_{1}^4}(\pi_{17+r+20})\cdot
M_{v_{2}^4}(\pi_{17+r+21})
\cdot \ldots \cdot 
M_{v_{s}^4}(\pi_{17+r+19+s})=
\begin{pmatrix} 
1&0&0&0&0&0&0&0&0&0\\ 0&1&0&0&0&0&0&0&0&0\\ 0&0&1&0&0&0&
 0&0&0&0\\ 0&0&0&1&0&0&0&0&0&0\\ 0&0&0&0&1&0&0&0&0&0\\ 0&0&0&0&0&1
 &0&0&0&0\\ 0&0&0&0&0&0&1&0&0&0\\ 0&0&0&0&0&0&0&1&0&0\\ 0&0&0&0&0&
 0&0&0&1&0\\ s&0&0&0&0&0&0&0&0&1\\ 
\end{pmatrix}.
$$

Finally, a direct computation gives

$$M_5=M_{v_{1}^5}(\pi_{17+r+20+s})\cdot
M_{v_{2}^5}(\pi_{17+r+21+s})
\cdot \ldots \cdot 
M_{v_{11}^5}(\pi_{17+r+30+s})=
\begin{pmatrix} 
1&1&0&0&0&0&0&0&0&0\\ 0&0&1&0&0&0&0&0&0&0\\ 0&0&0&1&0&0&
 0&0&0&0\\ 0&0&0&0&1&0&0&0&0&0\\ 0&0&0&0&0&1&0&0&0&0\\ 0&0&0&0&0&0
 &1&0&0&0\\ 0&0&0&0&0&0&0&1&1&1\\ 0&0&0&0&0&0&0&0&1&1\\ 1&1&1&1&1&
 1&1&1&1&0\\ 0&1&0&0&0&0&0&0&0&0\\
\end{pmatrix}.
$$

Now, we are in a position to compute the product of the above five matrices,

$$
N(r,s):=M_1M_2(r)M_3M_4(s)M_5
=
\begin{pmatrix} 
2 & 2 & 2 & 2 & 2 & 2 & 2 & 2 & 3 & 2 \\ 2\,s & 2\,s+2
  & 2 & 2 & 2 & 2 & 1 & 0 & 0 & 0 \\ 0 & 0 & 2 & 2 & 2 & 1 & 0 & 0 & 
 0 & 0 \\ 0 & 0 & 0 & r+2 & r+1 & 0 & 0 & 0 & 0 & 0 \\ 0 & 0 & 1 & 2
  & 2 & 0 & 0 & 0 & 0 & 0 \\ s & s+1 & 2 & 2 & 2 & 2 & 0 & 0 & 0 & 0
  \\ s+1 & s+2 & 2 & 2 & 2 & 2 & 2 & 0 & 0 & 0 \\ 0 & 0 & 0 & 0 & 0
  & 0 & 0 & 1 & 1 & 1 \\ 1 & 1 & 1 & 1 & 1 & 1 & 1 & 1 & 2 & 1 \\ 4\,
 s+2 & 4\,s+6 & 10 & r+13 & r+12 & 8 & 4 & 2 & 3 & 3 \\ 
\end{pmatrix}.
$$

When $r=s$ we adopt the notation 
\begin{equation}\label{Eq:M(r)}
M(r):=N(r,r).
 \end{equation}
Notice that $N(r,s)$ is the associate matrix to the path 
$$v(r,s)=v^1*v^2(r)*v^3*v^4(s)*v^5,$$
whose length is $48+s+r$. In particular, $v(r,r)$ is the corresponding path of $M(r)$ having length equal to $48+2r$. 


%

As a consequence of Lemmas~\ref{L:permutations} and \ref{L:permutations2} and the {previously} built matrices we immediately obtain:

\begin{proposition} \label{P:ciclo} Let $v=v(r,s)\in\{a,b\}^{48+r+s}$ for some $r,s\in\N$ as defined in (\ref{eqV}) and let $$\pi_0=(-3,-4,-5,-6,-7,-8,-9,10,1,-2).$$ Then:
\begin{enumerate}
 \item $v(\pi_0)=\pi_0$.
 \item $M_{v_1}(\pi_0)\cdot
M_{v_2}(\pi_1)\cdot \ldots \cdot M_{v_{48+r+s}}(\pi_{47+r+s})=N(r,s)$.
\end{enumerate}
\end{proposition}

\subsection{The path} \label{Subs:path}

\linero{The final path that we will follow, $u$, is the  \emph{concatenation} of an initial transition state $v^0$ in the Rauzy graph jointly with $w^1=v(r_1,r_1)$,$w^2=v(r_2,r_2)$, $\dots$, $w^k=v(r_k,r_k)$, $\dots$ for a suitable sequence of naturals $(r_k)_k$. Here, in turn, $v^0$ is meant the following concatenation:} 
$$v^0:=v(10^3,1)*v(10^3,10)*v(10^3,10^2)*v(10^3,10^3)*v(10^{2},10^{4})*v(10^2,10^5)*v(10^3,10^5)*v(10^7,10^7)*v(10^8,10^8),$$ \linero{whose} associate matrix is given by 
\begin{equation}\label{Eq:N0}
N_0:=M_{v^0}(\pi_0)
=N(10^3,1)\cdot N(10^3,10)\cdot N(10^3,10^2)\cdot M(10^3)\cdot N(10^2,10^4) \cdot N(10^2,10^5)\cdot N(10^3,10^{4})\cdot M(10^7)\cdot M(10^8).
\end{equation}

Then, our chosen path is 
\begin{equation}
u=(u_j)_j\in\{a,b\}^{\N},  \quad 
u=v^0*w^1*w^2*\dots *w^k*\dots,\label{E:chosenpath}
\end{equation}
where $w^{{k}}=v(r_{{k}},r_{{k}})$ for a suitable sequence of positive integers $\left(r_{{k}}\right)_{{k}}$ whose choice will be explained later 
in order to hold appropriate properties. 

In the following we must analyze the associate matrix  to the mentioned path. 
From $N_0$ (see (\ref{Eq:N0})) and (\ref{E:chosenpath}), we define 
\begin{eqnarray*}
N_1&:=&N_0\cdot M(r_1)=M_{v^0*w^1}(\pi_{0}),\\ 
N_2&:=&N_0\cdot M(r_1)\cdot M(r_2)=M_{v^0*w^1*w^2}(\pi_{0}),\\ 
& & \dots, \\
N_{k}&:=&N_0\cdot M(r_1)\cdot M(r_2)\cdot \ldots \cdot M(r_k)=M_{v^0*w^1*w^2*\dots*w^k}(\pi_{0}), \ldots
\end{eqnarray*}
Notice that \begin{equation}\label{Eq:Nk}
N_{k}=N_{k-1}\cdot M(r_k) \,\, \mathrm{for\, all\, }  k\geq 1.
\end{equation}

From now on, we deserve the letter $d$ to denote the dimension of the Euclidean space, and we use $n$ for denoting an arbitrary index.  In what follows, we will denote the matrix $N_n$ by 
$$c(n)=\left(\begin{array}{c c c c c c}
c_1(n)&\, c_2(n)&\, \dots&\, c_\ell(n)& \dots&\, c_{10}(n)
\end{array}\right),$$ 
{where,} for any $\ell \in\{1,2,\dots,10\}$ and for any $n\in\N\cup\{0\}$, $c_\ell(n)$ {denotes} the $\ell $-th
 column of the matrix $N_n$. We introduce now some useful notation.  {For} $v,w\in\R^{{d}}$, 
$K\in\R$   we write 
 $\frac{v}{w}$ to denote the vector made of the quotient of the corresponding components. Also we will say 
 $v<w$ when $v_j<w_j$ for any $j\in\{1,2,\dots, {{d}}\}$. 
 Let $V=\{v_i\}_{i=1}^k\subset\R^{{d}}$ then 
 $\max (V):=(m_j)\in\R^d$  with $m_j=\max\{(v_l)_j: 1\leq l\leq k\}$, analogously 
 $\min (V):=(m_j)\in\R^d$ with $m_j=\min\{(v_l)_j: 1\leq l\leq k\}$. 
 In $\mathbb{R}^{{d}}$, we will use the norm $|x|=\left\|x\right\|_0=\max\{\left|x_j\right|:1\leq j\leq {d}\}$ for a vector $x\in\R^{{d}}$.
Moreover sometimes we will need the standard Euclidean norm, and then we will write $\left\|\cdot\right\|_{e}$ to denote it; realize that 
 \begin{equation}\label{E:relnormas}
{\left|\cdot\right|\leq \left\|\cdot\right\|_e\leq \sqrt{d}\left|\cdot\right|.}  
 \end{equation}
 Also
$\left\langle v,w\right\rangle$ denotes the usual inner product and when writing $v<K${, for some $K\in\mathbb R$,}  we mean $|v|<K$.

 \begin{claim}\label{C:multiplicacionmatrices}
  Let $A=(a_{i,j}),B=(b_{i,j})$ be ${d}\times {d}$ real matrices and let $C=AB$. Denote by 
  $a_i$ and $c_i$ the $i$-th column of the matrices $A$ and $C$ respectively, $1\leq i\leq {d}$. Then $c_i=\sum_{j=1}^{{d}} a_jb_{j,i}.$
 \end{claim}

\begin{claim}\label{C:relationship}
 According to~(\ref{Eq:M(r)})-(\ref{Eq:Nk}) and taking into account Claim~\ref{C:multiplicacionmatrices}, the relationships among the columns of $N_{n+1}$ and $N_n$  are given by:
 \begin{eqnarray*}
  c_1(n+1)&=&2c_1(n)+2\rnMuno c_2(n)+\rnMuno c_6(n)+(\rnMuno +1)c_7(n)+c_9(n)+(4\rnMuno +2)c_{10}(n){,} \\
  c_2(n+1)&=&2c_1(n)+(2\rnMuno +2)c_2(n)+(\rnMuno +1)c_6(n)+(\rnMuno +2)c_7(n)+c_9(n)+(4\rnMuno +6)c_{10}(n){,}\\
  c_3(n+1)&=&2c_1(n)+2c_2(n)+2c_3(n)+c_5(n)+2c_6(n)+2c_7(n)+c_9(n)+10c_{10}(n){,}\\
  c_4(n+1)&=&2c_1(n)+2c_2(n)+2c_3(n)+(\rnMuno +2)c_4(n)+2c_5(n)+2c_6(n)+2c_7(n)+c_9(n)+(\rnMuno +13)c_{10}(n){,}\\
  c_5(n+1)&=&2c_1(n)+2c_2(n)+2c_3(n)+(\rnMuno +1)c_4(n)+2c_5(n)+2c_6(n)+2c_7(n)+c_9(n)+(\rnMuno +12)c_{10}(n){,}\\
  c_6(n+1)&=&2c_1(n)+2c_2(n)+c_3(n)+2c_6(n)+2c_7(n)+c_9(n)+8c_{10}(n){,}\\
  c_7(n+1)&=&2c_1(n)+c_2(n)+2c_7(n)+c_9(n)+4c_{10}(n){,}\\
  c_8(n+1)&=&2c_1(n)+c_8(n)+c_9(n)+2c_{10}(n){,}\\
  c_9(n+1)&=&3c_1(n)+c_8(n)+2c_9(n)+3c_{10}(n){,}\\
  c_{10}(n+1)&=&2c_1(n)+c_8(n)+c_9(n)+3c_{10}(n){.}\\
 \end{eqnarray*}

\end{claim}

\section{Relationships on the matrices associated to the path}\label{S:matricespath}

Our first result shows some useful properties of $c(0)=N_0$. Given $1\leq i,j\leq 10$, in general for any $n\geq 0$, by $\alpha_{i,j}(n)$ we will denote the angle between the column vectors $c_i(n)$ and $c_j(n)$. 
\begin{lemma}\label{L:valoresn=1} The initial matrix 
$$c(0)=N_0=N(10^3,1)\cdot N(10^3,10)\cdot N(10^3,10^2)\cdot M(10^3)\cdot N(10^2,10^4) \cdot N(10^2,10^5)\cdot N(10^3,10^{5})\cdot M(10^7) \cdot M(10^8)$$ 
verifies the following properties:  
 \begin{enumerate}
  \item $\alpha_{2,4}(0)\approx 0.613150240$ radians $\approx  35.130920977$ sexagesimal degrees.   
  \item $8.8\cdot 10^{-6}<\frac{c_3(0)}{c_4(0)}< 0.00256$.
  \item $c_3(0)<\min\{c_1(0),c_2(0),c_4(0),c_5(0)\}.$
  \item $c_2(0)>c_1(0)>c_4(0)>c_5(0)>\max\{c_3(0),c_6(0),c_7(0),c_8(0),c_9(0),c_{10}(0)\}$. In fact, 
	$$c_2(0)>c_1(0)>c_4(0)>c_5(0)> c_3(0)>c_6(0)>c_7(0)>c_9(0)>c_{10}(0)>c_{8}(0).$$
 \end{enumerate}
\end{lemma}
\begin{proof}
Using a mathematical software, for instance \emph{Maxima}, the columns of the matrix $N_0$  
are tabulated as follows:

 %
 %
%

\newcommand{\tabulavectorgrande}[1]{
\noindent\begin{flushleft}
\begin{minipage}{\textwidth}
{\tiny 							
$\begin{aligned}
#1
\end{aligned}
$
}
\end{minipage}
\end{flushleft}
}

\tabulavectorgrande{
c_1(0)&=&( 336229277950011717660363178542095919820, 201967636924785649494795541341599531256, 17589953114807902467073230308638170012,\\
& & 4904834498292079642252074333015146057,673183649762255347016600641516654502,108427847586847325538107333139460307226,\\
& & 189066665339392901941522713202016181756,86784985800059461041522755435882101903,170839083866711572727387924392756316627,\\
& & 858859399063899832091130671508341910628)
}
\tabulavectorgrande{
c_2(0)&=& (336229281312304463537560461534387268964,201967638944461998545893090051538616554,17589953290707431856157356862572763364,\\
& & 4904834547340424134689536190075075729,673183656494091777320804955276717798,108427848671125790563798377890799761165,\\
& & 189066667230059536428789641745600455164,86784986667909310363620825292794324633,170839085575102394310599279714863485548,\\
& & 858859407652493736844209269230250658737)
}
\tabulavectorgrande{
c_3(0)&=&(6724612541543264956966751301682,4039368901385059344525956263782,351800671037916200638568783464,\\
& & 98207197706162797133942935065,13463946147438001128999534448,2168565730929434454750015698994,\\
& & 3781348451471746347423755096979,1735706677027050139353864490993,3416795383741574444420404374640,\\
& & 17177367440221022102567989614413)
}
\tabulavectorgrande{c_4(0)&=&(2636051966399736281453632722746282,1578924969986525211963142074097802,157620037474770709310050412223976,\\
& & 11049865976008017853785535419425061,27190679154735633279369625993803,857860122243592887089494724048810,\\
& & 1479491100761897064155676215897294,680047384327032904778733810129953,1339035544331075575245509796617848,\\
& & 17787004852029275641036901194433027)
}
\tabulavectorgrande{
c_5(0)&=&( 2636051940106463531669859262374434,1578924954237669673577391253586336,157620035902088388453190560538887,\\
& & 11049865865510343480612539214458009,27190678882963489346651720430674,857860113686677578670311056006734,\\
& & 1479491086004799983764320976064204,680047377543916331772165800996793,1339035530974888486457769352316008,\\
& & 17787004674331006126094264043459277)
}
\tabulavectorgrande{c_6(0)&=&(6724585743926472996397847612526,4039352849577990339359848868979,351799071970659099150771547692,\\
&  & 98096692674524976053834794395, 13463673365518320957593200270,2168557011372314256042313207270,\\
& & 3781333410774597518878804178730,1735699763732966103997409362065,3416781771295796216168493281808,\\
& & 17177188453662077506841013409861)
}
\tabulavectorgrande{c_7(0)&=&(5043438320268470857714120596625,3029514043901832763157527773346,263849252307341555920443013596,\\
& & 73572505089616924687763638299,10097753046639519292302676590,1626417440021498832982334488332,\\
& & 2835999502696001394247010397680,1301774567868085712895794679869,2562585826630575275650498329534,\\
& & 12882888817331401848001502584467)
                   }
\tabulavectorgrande{
c_8(0)&=&( 362291905303653610244790047886,2019675844131795723935335346578,175899485414183072428641652150,\\
& & 49048332230335163458058741433,6731834747344359740974171084,1084278193955959588943286872674,\\
& & 1890666161820407448765023586346,867849632359543540999822440384,1708390394485324463746603092522,\\
& & 8588591757603678967516110418151)
}
\tabulavectorgrande{c_9(0)&=&( 5043437899983169979222189979768,3029513791443065429713378397494,263849230319967992643489472546,\\
& & 73572498958592878225045696659,10097752205162553326997484020,1626417304487107466321173135306,\\
& & 2835999266363398794161600789578,1301774459387188120931320227721,2562585613082379227704206267786,\\
& & 12882887743760465093613874955371)
}
\tabulavectorgrande{c_{10}(0)&=&(3362292073414832120494672496670,2019675945113535539640918622302,175899494208978594209301791322,\\
& & 49048334682701825681593837472,6731835083929256018066697278,1084278248168767445132915746674,\\
& & 1890666256351794248289714412662,867849675751143251758826240245,1708390479903108124797461126868,\\
& & 8588592187024539029467183604381)
}

To see at first glance the magnitude of matrix $N_0$, we write it rounding all positions to two decimals in the mantissa:

{\tiny
\begin{equation}\label{E:c0}
\begin{pmatrix} 3.36  \cdot 10^{38} & 3.36\cdot 10^{38} &   6.72 \cdot 10^{30} & 2.64 \cdot 10^{33} & 2.64 \cdot 10^{33} &  
6.72 \cdot 10^{30} & 5.04\cdot 10^{30}& 3.36\cdot 10^{30} & 5.04\cdot 10^{30} & 3.36\cdot 10^{30} \\ 
	2.02  \cdot 10^{38} 
  &  2.02  \cdot 10^{38} & 4.04 \cdot 10^{30} & 1.58 \cdot 
 10^{33} &  1.58 \cdot 10^{33}  & 4.04 \cdot 10^{30} &  3.03 \cdot 10^{30} &
2.02  \cdot 10^{30} & 3.03  \cdot 10^{30} & 2.02
   \cdot 10^{30} \\ 
	1.76 \cdot 10^{37} &  1.76 \cdot 10^{37} & 3.52\cdot 10^{29} & 1.58  \cdot 10^{32} &1.58 \cdot 10^{32} & 3.52
   \cdot 10^{29} & 2.64  \cdot 10^{29} & 1.76 \cdot 10^{29} & 2.64
   \cdot 10^{29} & 1.76  \cdot 10^{29} \\ 	
 4.90\cdot 10^{36}&   4.90  \cdot 10^{36} & 9.82 \cdot 10^{28} &  1.10 \cdot 10^{34} & 1.10
   \cdot 10^{34} & 9.81 \cdot 10^{28} & 7.36  \cdot 10^{28} & 4.90
   \cdot 10^{28} & 7.36  \cdot 10^{28} & 4.90\cdot 10^{28}
  \\ 
6.73 \cdot 10^{35} 	& 6.73  \cdot 10^{35}  & 1.35  \cdot 10^{28}
  & 2.72  \cdot 10^{31} & 2.72  \cdot 10^{31} & 1.35 \cdot 10^{28}
  & 1.01  \cdot 10^{28} & 6.73  \cdot 10^{27} & 1.01  \cdot 10^{28}
  & 6.73 \cdot 10^{27} \\ 
1.08 \cdot 10^{38} & 1.08   \cdot 10^{38} & 2.17  \cdot 10^{30} & 8.58  \cdot 10^{32} & 8.58 
 \cdot 10^{32} & 2.17  \cdot 10^{30} & 1.63  \cdot 10^{30} & 1.08
   \cdot 10^{30} & 1.63 \cdot 10^{30} & 1.08 \cdot 10^{30}
  \\  	
 1.89\cdot 10^{38}&  1.89 \cdot 10^{38} & 3.78 \cdot 
 10^{30} & 1.48 \cdot 10^{33} & 1.48  \cdot 10^{33} & 3.78
 \cdot 10^{30} & 2.84 \cdot 10^{30} & 1.89  \cdot 10^{30} & 2.84 
   \cdot 10^{30} &1.89 \cdot 10^{30} \\ 	
8.68  \cdot 10^{37}	 & 8.68  \cdot 10^{37}	& 1.74  \cdot 10^{30} & 6.80  \cdot 10^{32} & 6.80
   \cdot 10^{32} & 1.74  \cdot 10^{30} & 1.30  \cdot 10^{30} & 8.68
   \cdot 10^{29} & 1.30 \cdot 10^{30} & 8.68\cdot 10^{29} \\ 
	1.71
   \cdot 10^{38}&  1.71\cdot 10^{38} & 3.42 \cdot 10^{30} &1.34
   \cdot 10^{33} & 1.34\cdot 10^{33} & 3.42  \cdot 10^{30} & 2.56
   \cdot 10^{30} & 1.71 \cdot 10^{30} & 2.56  \cdot 10^{30} & 1.71
   \cdot 10^{30} \\
		8.59\cdot 10^{38} & 8.59 \cdot 10^{38}
  & 1.72  \cdot 10^{31} & 1.78  \cdot 10^{34} & 1.78\cdot 
 10^{34} & 1.72 \cdot 10^{31} & 1.29 \cdot 10^{31} & 8.59
 \cdot 10^{30} & 1.29 \cdot 10^{31} & 8.59\cdot 10^{30} \\ 
 \end{pmatrix} 
\end{equation}

}

Concerning the angle $\alpha_{2,4}(0)$ between the columns $c_2(0)$ and $c_4(0)$, we find 
{\tiny
\begin{eqnarray*}
& & \cos(\alpha_{2,4}(0))=\frac{\left\langle c_2(0), c_4(0)\right\rangle}{\left\|c_2(0)\right\|_e \left\|c_4(0)\right\|_e}\\
&=&\frac{  17199250545610936768824152621451389334681231597441086451563040915278861401 }{
4\sqrt{ 113281183918496857424656165570234490385231936299604045194111322953603 }}\\
& & \times\frac{1}{\sqrt{ 244008721509185038750928385963528788972842561456599203452853614664788080554809}}\\
&=&   0.8178392835711894...,
\end{eqnarray*}
}
hence $\alpha_{2,4}(0)\approx  0.6131502403070084$ radians, or $\alpha_{2,4}(0)\approx   35.13092097702379$ degrees. This proves Part (1). The other inequalities 
are easily obtained from the values of the columns $c_j(0)$ of $N_0$.
\end{proof}

Using the relationships given by Claim~\ref{C:relationship} we obtain:

\begin{lemma} \label{C:ordencolumnas}
Let $(r_n)_{n\geq 1}$ be an increasing sequence, with $r_n>8$. For any $n\in\N\cup\{0\}$ it holds:
$$\cdosn>\cunon> \ccuatron>\ccincon>\max\{c_j(n)\}_{j\in\{3,6,7,8,9,10\}},$$
$$\cdosn>\cunon>\ccuatron>\ccincon>\ctresn>\cseisn>\csieten>\cnueven>\cdiezn>\cochon.$$
\end{lemma}

\begin{proof}
Case $n=0$ is immediate from the values of $c(0)$ presented in Lemma~\ref{L:valoresn=1}.   

By simply inspecting the relationships given in Claim~\ref{C:relationship}, we obtain
 $c_2(n)>c_1(n)$, $c_4(n)>c_5(n)$, $c_5(n)>c_3(n)$, $c_3(n)>c_6(n)>c_7(n)$, $c_9(n)>c_{10}(n)>c_{8}(n)$ for all $n\geq 0$.  

It only remains to prove that $c_1(n)>c_4(n)$ and $c_7(n)>c_9(n)$ for $n\geq 1$. We apply induction, by assuming that the hypothesis 
of the statement are true for $m<n$. Then,  
\begin{eqnarray*}
  c_7(n)-c_9(n)&=&-\cunonmuno+\cdosnmuno+2\csietenmuno-\cochonmuno-\cnuevenmuno+\cdieznmuno\\
  &>&[\cdosnmuno-\cunonmuno]+2\csietenmuno-\csietenmuno-\csietenmuno+\cdieznmuno
  \\
  &=&[\cdosnmuno-\cunonmuno]+\cdieznmuno>0;
 \end{eqnarray*}
 
\begin{eqnarray*}
  \cunon-\ccuatron&=&(2\rn-2)\cdosnmuno-2\ctresnmuno-(\rn+2)\ccuatronmuno-2\ccinconmuno\\
	& & +(\rn-2)\cseisnmuno+(\rn-1)\csietenmuno+(3\rn-11)\cdieznmuno
  \\
  &>&(2\rn-2)\cdosnmuno-2\cdosnmuno-(\rn+2)\cdosnmuno-2\cdosnmuno\\
	& & +(\rn-2)\cseisnmuno+(\rn-1)\csietenmuno+(3\rn-11)\cdieznmuno\\
  &=& (\rn-8)\cdosnmuno+(\rn-2)\cseisnmuno+
  (\rn-1)\csietenmuno+(3\rn-11)\cdieznmuno
  >0{,}
 \end{eqnarray*}
if $\rn>8$.
\end{proof}

\subsection{Relationship between $\cdosn$ and $\cdosnMuno$}
 
Our interest in this subsection is to prove Theorem~\ref{T:c2}, in which we will give an estimate of the ratio $\frac{\cdosnMuno}{\cdosn}$  
in terms of appropriate sequences $(r_n)$ and $(p_n)$. 
We introduce some technical lemmas before. The proof of the first one is immediate and we omit its proof.

\begin{lemma}\label{L:herigon}
 Let $a,b,c,d,p,q$ be positive real numbers such that $\frac{a}{b}<\frac{c}{d}$. Then:
 \begin{enumerate}
  \item $\frac{a}{b}<\frac{a+c}{b+d}<\frac{c}{d}$.
  \item If $\frac{a}b>\frac{p}q$ and $\frac{c}d>\frac{p}q$ then $\frac{a+c}{b+d}>\frac{p}q$ (the same is true reversing the inequalities).
 \end{enumerate}

\end{lemma}

For the second lemma, recall that $r_n>0$ for all $n\in\mathbb N$ and that $\frac{c_i(n)}{c_j(n)}$ is meant a componentwise division 
of $i$-th and $j$-th columns of $N_n$.  Also, for $K\in\mathbb R,$ $\frac{c_i(n)}{c_j(n)}\geq K$ is used to indicate that any  element of the componentwise division is greater than or equal to $K$. Notice that all the matrices $N_n$, and consequently all their columns $c_j(n)$, are positive. 

\begin{lemma}\label{L:cocientec2c6}
 For any $n\in\N\cup\{0\}$ we have:  
 $$\frac{\cdosnMuno}{\cseisnMuno}\geq \frac{\rnMuno +1}3{.}$$
\end{lemma}
\begin{proof}

We use Claim~\ref{C:relationship} and Lemma~\ref{C:ordencolumnas}:
\begin{eqnarray}
 \frac{\cdosnMuno}{\cseisnMuno}&=&\frac{\defcdosnMuno}{\defcseisnMuno}\nonumber \\
 &\geq&\frac{(2\rnMuno+2)\cdosn+(\rnMuno+3)\cseisn+(\rnMuno+2)\csieten+(4\rnMuno+7)\cdiezn}{\defcseisnMuno}\nonumber\\
 &\geq&\frac{(2\rnMuno+2)\cdosn+(\rnMuno+3)\cseisn+(\rnMuno+2)\csieten+(4\rnMuno+7)\cdiezn}{
 5\cdosn+2\cseisn+3\csieten+8\cdiezn}\label{E:c2c60}
\end{eqnarray}

Now observe that $\frac{(4\rnMuno+7)\cdiezn}{8\cdiezn}\geq\frac{\rnMuno+1}{3}$ and 
$\frac{(\rnMuno+2)\csieten}{3\csieten}\geq\frac{\rnMuno+1}{3}$ and then by Lemma~\ref{L:herigon} we have:
\begin{equation}
 \frac{(\rnMuno+2)\csieten+(4\rnMuno+7)\cdiezn}{3\csieten+8\cdiezn}\geq\frac{\rnMuno+1}{3}. \label{E:c2c6A}
\end{equation}

Also, $\frac{(2\rnMuno+2)\cdosn}{5\cdosn}\geq \frac{\rnMuno+1}{3}$  and 
$\frac{(\rnMuno+3)\cseisn}{2\cseisn}\geq\frac{\rnMuno+1}{3}$ and then by Lemma~\ref{L:herigon} we have:
\begin{equation}
 \frac{(2\rnMuno+2)\cdosn+(\rnMuno+3)\cseisn}{5\cdosn+2\cseisn}\geq\frac{\rnMuno+1}{3}. \label{E:c2c6B}
\end{equation}

Applying Lemma~\ref{L:herigon} to Equations~(\ref{E:c2c6A}) and (\ref{E:c2c6B}), and taking into account
Equation~(\ref{E:c2c60}), we obtain:
$$\frac{\cdosnMuno}{\cseisnMuno}\geq\frac{\rnMuno+1}3.$$

\end{proof}


 \begin{theorem}\label{T:c2} Let $(p_n)_n$ be a strictly increasing sequence of naturals. 
Then it is possible to choose a strictly  increasing sequence $(\rn)_n$ such that for any $n\in\N \cup\{0\}$ 
 we have
 $$2\rnMuno\cdosn<\cdosnMuno<2\rnMuno\cdosn(1+10^{-\pnMuno}).$$
In fact, we can take $r_n=10^{k+p_{n+1}}$, being $k\geq 2$ constant. 
\end{theorem}
\begin{proof}
%

 We recall (see Claim~\ref{C:relationship}) that for any $n\in\N \cup\{0\}$:
 $$c_2(n+1)=2c_1(n)+(2\rnMuno +2)c_2(n)+(\rnMuno +1)c_6(n)+(\rnMuno +2)c_7(n)+c_9(n)+(4\rnMuno +6)c_{10}(n).$$
 Thus, it is evident that $\cdosnMuno>2\rnMuno\cdosn$.
 Also, by using the inequalities from Lemma~\ref{C:ordencolumnas} we  have 
 \begin{eqnarray*}
c_2(n+1)&=&2c_1(n)+(2\rnMuno +2)c_2(n)+(\rnMuno +1)c_6(n)+(\rnMuno +2)c_7(n)+c_9(n)+(4\rnMuno +6)c_{10}(n)  \\
&\leq & (2\rnMuno +4)c_2(n)+(6\rnMuno +10)c_6(n).
 \end{eqnarray*}
We need now to show that $(2\rnMuno +4)c_2(n)+(6\rnMuno +10)c_6(n)<[2\rnMuno+2\rnMuno 10^{-\pnMuno}]\cdosn$ which is equivalent to prove 
$$\frac{6\rnMuno+10}{2\rnMuno10^{^{-\pnMuno}}-4}<\frac{\cdosn}{\cseisn}.$$
By using Lemma~\ref{L:cocientec2c6},  we 
have 
$$\frac{\cdosn}{\cseisn}\geq\frac{\rn+1}{3}.$$
Then, it  will be enough if we obtain that
\begin{equation}
\frac{3\rnMuno+5}{ \rnMuno10^{^{-\pnMuno}}-2}<\frac{\rn+1}{3}.\label{E:aprobarc2} 
\end{equation}
Since $\lim_{x\to\infty} \frac{3x+5}{x10^{-\pnMuno}-2}=3\cdot 10^{\pnMuno}$
then we can guarantee~(\ref{E:aprobarc2}) by taking $\rn$ and $\rnMuno$ big enough in order to satisfy  
$3\cdot 10^{\pnMuno}<\frac{\rn+1}3$, $3\cdot 10^{p_{n+2}}<\frac{\rnMuno+1}3$,
which is always possible by taking $\rnMuno$  big enough. For instance, this is easily achieved if we take
$r_n=10^{k+p_{n+1}}$ with $k\geq 2$ constant. Indeed, (\ref{E:aprobarc2}) is rewritten as 
$\frac{5+3\cdot 10^{{\tiny k+p_{n+2}}} }{ 10^{{\tiny k+p_{n+2}}} 10^{{\tiny -p_{n+1}}}-2}<\frac{1+10^{{\tiny k+p_{n+1}}}}{3},$ 
and the inequality holds if and only if 
$$9\cdot 10^{{\tiny k+p_{n+2}}} + 2\cdot 10^{{\tiny k+p_{n+1}}}+17 < 10^{{\tiny 2k+p_{n+2}}}+ 10^{{\tiny k+p_{n+2}-p_{n+1}}}, $$
which is satisfied due to $17<10^{{\tiny k+p_{n+2}-p_{n+1}}}$ and 
$9\cdot 10^{{\tiny k+p_{n+2}}} + 2\cdot 10^{{\tiny k+p_{n+1}}}< 11\cdot 10^{{\tiny k+p_{n+2}}} < 10^{{\tiny 2k+p_{n+2}}}$ 
because $p_{n+2}-p_{n+1}\geq 1$ and $k\geq 2$.

 
\end{proof}

 \begin{coro}\label{C:c1} Let $(p_n)_n$ be a strictly increasing sequence of naturals. 
Then it is possible to choose a strictly  increasing sequence $(\rn)_n$ such that for any $n\in\N \cup\{0\}$ 
 we have
 $$2\rnMuno\cdosn<\cunonMuno<2\rnMuno\cdosn(1+10^{-\pnMuno}).$$
In fact, we can take $r_n=10^{k+p_{n+1}}$, being $k\geq 2$ constant. 
\end{coro}
\begin{proof}
By Claim~\ref{C:relationship} we have for any $n\in\N \cup\{0\}$:
 $$c_1(n+1)=2c_1(n)+2\rnMuno c_2(n)+\rnMuno c_6(n)+(\rnMuno +1)c_7(n)+c_9(n)+(4\rnMuno +2)c_{10}(n).$$ 
 Thus, it is clear that $\cunonMuno>2\rnMuno\cdosn$.
 Also, by using the inequalities from Lemma~\ref{C:ordencolumnas} we  have 
 \begin{eqnarray*}
c_1(n+1)&=&2c_1(n)+2\rnMuno c_2(n)+\rnMuno c_6(n)+(\rnMuno +1)c_7(n)+c_9(n)+(4\rnMuno +2)c_{10}(n)  \\
&\leq & (2\rnMuno +2)c_2(n)+(6\rnMuno +10)c_6(n)\\
&\leq & (2\rnMuno +4)c_2(n)+(6\rnMuno +10)c_6(n).
 \end{eqnarray*}
\linero{Since $\cunonMuno\leq\cdosnMuno$, from this point} the proof of Theorem~\ref{T:c2} applies.
\end{proof}

\subsection{Relationship between $\ccuatron$ and $\ccuatronMuno$}
 
Our interest now is to prove, by recurrence, Theorem~\ref{T:c4} about the existence of an increasing sequence $\rn$ such that 
	$$\rn\ccuatronmuno<\ccuatron<\rn\ccuatronmuno(1+10^{-\pn}).$$
We begin with   a preliminary result. 

\begin{lemma}
\label{L:cocientec34} Let $(p_n)_n$ be a 
sequence of positive numbers with $p_1=2$. 
Then, for each $n\in\N$ it is possible to choose $\rn$ big enough, fixing $r_1=10^{10}$, such that 
\begin{equation}\label{Eq:trescuartos} 
\frac{\cjotan}{\celen}<\frac{1}{20}10^{-\pn},\quad j\in\{3,6,7,8,9,10\},
\; l\in\{1,2,4,5\}.
\end{equation}
\end{lemma}
\begin{proof}
%
%
Observe that, by Lemma~\ref{C:ordencolumnas}, for any   $j\in\{3,6,7,8,9,10\}$ and $l\in\{1,2,4,5\}$ we have:
$$\frac\cjotan\celen<\frac\ctresn\ccincon.$$ 
Then we will be done if we show that $\frac\ctresn\ccincon<\frac{1}{20}10^{-\pn}$.

\def\defctresnuno{2c_1(0)+2c_2(0)+2c_3(0)+c_5(0)+2c_6(0)+2c_7(0)+c_9(0)+10c_{10}(0)}
\def\defccincouno{2c_1(0)+2c_2(0)+2c_3(0)+(r_1 +1)c_4(0)+2c_5(0)+2c_6(0)+2c_7(0)+c_9(0)+(r_1 +12)c_{10}(0)}

Realize that, according to Claim~\ref{C:relationship}, the quotient $\frac{\ctresnMuno}{\ccinconMuno}$ equals to
 $$
 \frac{\defctresnMuno}{\defccinconMuno}.$$
Then note that, for $n\in\mathbb{N}$, it is possible to choose $\rnMuno$ big enough to obtain 
 $$\frac{\ctresnMuno}{\ccinconMuno}<\frac1{20}10^{-\pnMuno}.$$
Finally we need to prove that, for $n=0$, with  $r_1=10^{10}$, the value $\frac{c_3(1)}{c_5(1)}$  verifies the corresponding bound, that is:
 $$\frac\defctresnuno\defccincouno<\frac1{20}10^{-2}.$$
In order to prove this inequality   it is necessary to use  the values of 
$c(0)$ given in the proof of Lemma~\ref{L:valoresn=1}.  With a simple use of Lemma~\ref{L:herigon}-(2)  we conclude the proof by considering (we also apply  Lemma~\ref{C:ordencolumnas}):
$$\frac{p}{q}=\frac1{20}10^{-2},$$ 
$$\frac{a}{b}=\frac{2c_1(0)+2c_2(0)+c_5(0)}{
2c_1(0)+2c_2(0)+(r_1 +1)c_4(0)+2c_5(0)}<\frac{5c_2(0)}{(7+r_1)c_5(0)}<\frac{p}{q},$$
$$\frac{c}{d}=\frac{2c_3(0)+2c_6(0)+2c_7(0)+c_9(0)+10c_{10}(0)}{
2c_3(0)+2c_6(0)+2c_7(0)+c_9(0)+(r_1 +12)c_{10}(0)}
<\frac{17c_3(0)}{(19+r_1)c_{10}(0)}<\frac{p}{q}.$$

\end{proof}

\begin{theorem}\label{T:c4} Let $(p_n)_n$ be the sequence $p_n=n+1$,  $n\geq 1$. Then, there exists an increasing sequence 
$(r_n)_{n\geq 1}$ of positive numbers such that  
\begin{equation}\label{E:propic4}
r_{n+1}\ccuatron <\ccuatronMuno<r_{n+1} \ccuatron (1+10^{-p_{n+1}})
\end{equation} 
for any $n\in\N\cup\{0\}$. 
\end{theorem}
\begin{proof}
 We use recurrence. For the first step, we need to prove the existence of a positive $r_1$ such that $c(0)=N_0$ and $c(1)=c(0)\cdot M(r_1)=N_1$  (recall the definitions of $c(1)$ and $M(r)$ in Subsection~\ref{Subs:path})  verify
$$r_1<\frac{c_4(1)}{c_4(0)}<r_1\left(1+10^{-2}\right)=1.01\, r_1.$$
Using a mathematical software, we find {\small
\begin{eqnarray*}
  c_4(1)&=&( 2639414258473151113574127395242952\cdot r_1+1344927748470953440647878694947308077144,\\
	&& 1580944645931638747502782992720104\cdot r_1+807876918940116412119940682984917963530,\\
  && 157795936968979687904259714015298\cdot r_1+70360447776618068033180129608346841714,\\
	&& 11049915024342700555611217013262533 \cdot r_1
+19663538805901685050878025032296529025,\\
  && 27197410989819562535387692691081 \cdot r_1+2692843546891123062354627555934964804,\\
	&& 858944400491761654534627639795484\cdot r_1+
433714851605532988072322282270484879130,\\
  && 1481381767018248858403965930309956 \cdot r_1+756272631315301738956135092916644697798,\\
	&& 680915234002784048030492636370198 \cdot r_1
+347142687255643546117014569546773523324,\\
  && 1340743934810978683370307257744716 \cdot r_1+683361738589766361162774537323530361096,\\
&& 17795593444216300180066368378037408 \cdot r_1+3435508980461315455942689535315933953144)
\end{eqnarray*} 
}
and 
{\small
\begin{eqnarray*}
\frac{c_4(1)}{c_4(0)}&=&( 1.00127550295528 \cdot r_1+510205.3243312298, 1.001279146244125\cdot r_1+511662.6402754344,\\
& & 1.001115971655806\cdot r_1+446392.7867539065, 1.000004438817158 \cdot r_1+1779.527357942266,\\
& & 1.000247578776742 \cdot r_1+99035.53830218055, 1.001263933618144 \cdot r_1+505577.5881867812,\\
& & 1.001277916612934 \cdot r_1+511170.7876619482, 1.001276160596676 \cdot r_1+510468.3809631471,\\
& & 1.001275836543052 \cdot r_1+510338.7594771761, 1.000482857696305 \cdot r_1+193147.1323610375)
\end{eqnarray*} 
}
\linero{Obviously, $\frac{c_4(1)}{c_4(0)}>r_1.$ On the other hand,} notice that $\frac{c_4(1)}{c_4(0)}\leq  1.001279146244125\cdot r_1+511662.6402754344$. In particular, for $r_1=10^{10}$ we have 
$\frac{c_4(1)}{c_4(0)}\leq  1.001330312508153\cdot 10^{10} < 10^{10}\cdot (1+10^{-2}),$ 
so we have finished the first induction step. 


In order to prove~(\ref{E:propic4}), 
recall that, by Claim~\ref{C:relationship}, 
  \begin{eqnarray*}
	& & \ccuatronMuno\\
&=&	\defccuatronMuno.
\end{eqnarray*} 
  Then it is clear that $\ccuatronMuno>\rnMuno\ccuatron.$ 
We now prove the upper inequality in 
  Equation~(\ref{E:propic4}). Using Lemma~\ref{C:ordencolumnas}, we have  
  $$\ccuatronMuno\leq 8\cdosn+(\rnMuno+20)\ctresn+\rnMuno \ccuatron,
  $$
 then we will finish the proof  if we show that 
 $8\cdosn+(\rnMuno+20)\ctresn+\rnMuno \ccuatron<\rnMuno\ccuatron(1+10^{-\pnMuno})$ or simply:
 \begin{equation}\label{E:c4aprobar}
  8\cdosn+(\rnMuno+20)\ctresn<\rnMuno\ccuatron10^{-\pnMuno}.
 \end{equation}

 Observe that we can choose $\rnMuno$ big enough to satisfy $ 8\cdosn<\frac\rnMuno{20}10^{-\pnMuno}\ccuatron$
 and then by Lemma~\ref{L:cocientec34}
 \begin{equation}\label{Eq:papas26}
  8\cdosn+(\rnMuno+20)\ctresn<\frac\rnMuno{20}10^{-\pnMuno}\ccuatron
  +\frac{\rnMuno+20}{20}10^{-\pn}\ccuatron.
 \end{equation}
 \linero{Then,  inequality~(\ref{E:c4aprobar}) will occur if (multiply  by $10^{\pn}$ the right part of~(\ref{Eq:papas26}) and the corresponding right part of~(\ref{E:c4aprobar}),}  with $p_{n+1}-p_n=1$): 
 $$\frac\rnMuno{20}10^{-1}\ccuatron
  +\frac{\rnMuno+20}{20}\ccuatron<\rnMuno\ccuatron10^{-1},$$
equivalently:
 $$\frac\rnMuno{20}10^{-1}
  +\frac{\rnMuno+20}{20}<\rnMuno10^{-1},$$
or
 $$\frac\rnMuno{200}
  +\frac{\rnMuno}{20}+1<\frac\rnMuno{10},$$
 which holds for instance if $\rnMuno>25$. So, to finish the induction it suffices to consider a sufficiently large number $r_{n+1}$ with 
$r_{n+1}\geq\max\{r_n,25\}+1,$ and  $r_1=10^{10}$.
 \end{proof}

\nuevo{ 
\begin{coro}\label{C:c5} Let $(p_n)_n$ be the sequence $p_n=n+1$,  $n\geq 1$. Then, there exists an increasing sequence 
$(r_n)_{n\geq 1}$ of positive numbers such that  
\begin{equation}\label{E:propic5}
r_{n+1}\ccuatron <\ccinconMuno<r_{n+1} \ccuatron (1+10^{-p_{n+1}})
\end{equation} 
for any $n\in\N\cup\{0\}$. 
\end{coro}

\linero{
\begin{proof}
We start the proof by checking that the property is true in the first step, when comparing the columns $c_5(1)$ and $c_4(0)$. To this end, 
by using a mathematical software, we find {(recall that $c(1)=c(0)\cdot M(r_1)$)} {\small
\begin{eqnarray*}
  c_5(1)&=&(  2639414258473151113574127395242952\cdot r_1+1344925109056694967496765120819912834192,\\
	&&  1580944645931638747502782992720104\cdot r_1+807875337995470480481193180201925243426,\\
	&&  157795936968979687904259714015298\cdot r_1+70360289980681099053492225348632826416,\\
	&&  11049915024342700555611217013262533\cdot r_1+19652488890877342350322413815283266492,\\
  &&  27197410989819562535387692691081\cdot r_1 +2692816349480133242792092168242273723,\\	
	&& 858944400491761654534627639795484\cdot r_1+433713992661132496310667747642845083646,\\
  &&  1481381767018248858403965930309956\cdot r_1+756271149933534720707276688950714387842,\\
	&&  680915234002784048030492636370198\cdot r_1+347142006340409543332966539054137153126,\\
  &&  1340743934810978683370307257744716\cdot r_1+683360397845831550184091167016272616380,\\
&& 17795593444216300180066368378037408\cdot r_1+3435491184867871239642509468947555915736),
\end{eqnarray*} 
}
and taking into account the value  of $c_4(0)$ -see Lemma~\ref{L:valoresn=1}-, we obtain: 
{\small

\begin{eqnarray*}
\frac{c_5(1)}{c_4(0)}&=&( 1.00127550295528\cdot r_1 + 510204.3230557268,  1.001279146244125\cdot r_1 + 511661.6389962881,\\
& & 1.001115971655806\cdot r_1 + 446391.7856379349,  1.000004438817158\cdot r_1 + 1778.527353503449,\\
& &  1.000247578776742\cdot r_1 + 99034.53805460177,  1.001263933618144\cdot r_1 + 505576.5869228475,\\
& & 1.001277916612934\cdot r_1 + 511169.7863840316,  1.001276160596676\cdot r_1 + 510467.3796869864,\\
& &  1.001275836543052\cdot r_1 + 510337.7582013395, 1.000482857696305\cdot r_1 + 193146.1318781798 ).
\end{eqnarray*} 
} 
The inequality $\frac{c_5(1)}{c_4(0)}>r_1$ holds trivially; concerning the converse inequality, we find 
$\frac{c_5(1)}{c_4(0)}\leq   1.001279146244125\cdot r_1 + 511661.6389962881$. 
In particular, for $r_1=10^{10}$, we get 
$\frac{c_5(1)}{c_4(0)}\leq  1.001330312408025\cdot 10^{10} < 10^{10}\cdot (1+10^{-2}),$ which ends the first step of the induction.

By Claim~\ref{C:relationship} we have for any $n\in\N \cup\{0\}$:
 $$c_5(n+1)=2c_1(n)+2c_2(n)+2c_3(n)+(\rnMuno +1)c_4(n)+2c_5(n)+2c_6(n)+2c_7(n)+c_9(n)+(\rnMuno +12)c_{10}(n).
$$ 
 Thus $\ccinconMuno>\rnMuno c_4(n)$.
 Also, by using  Lemma~\ref{C:ordencolumnas} we  have 
 \begin{eqnarray*}
c_5(n+1)&=&2c_1(n)+2c_2(n)+2c_3(n)+(\rnMuno +1)c_4(n)+2c_5(n)+2c_6(n)+2c_7(n)+c_9(n)+(\rnMuno +12)c_{10}(n)\\
&\leq & 4c_2(n)+(\rnMuno +3)c_4(n)+(\rnMuno +19)c_3(n)\\
&\leq & 7c_2(n)+\rnMuno c_4(n)+(\rnMuno +19)c_3(n)\\
&\leq & 8c_2(n)+\rnMuno c_4(n)+(\rnMuno +20)c_3(n),
 \end{eqnarray*}
with $c_5(n+1)\leq c_4(n+1)$,
 and from here the proof of Theorem~\ref{T:c4} applies.
\end{proof}}

}

\subsection{The angle between $\cdosn$ and $\ccuatron$}

As introduced before Lemma~\ref{L:valoresn=1},  $\alpha_{i,j}(n)$ denotes the angle between columns $i$ and $j$
in $N_n$ \linero{(namely, between $c_i(n)$ and $c_j(n)$)}. Our objective in this section is to show that the second and fourth columns do not tend to the same direction
as $n$ increases. Following Theorem~\ref{T:c2} and Theorem~\ref{T:c4} our final choice of the values $p_n,  n\geq 1$,  is:
\begin{eqnarray*}
p_n&=&n+1.
\end{eqnarray*}
\begin{remark}\label{R:choicesrp}
With these values of $p_n$,  we can apply the above mentioned theorems by starting with $r_1=10^{10}$: for Theorem~\ref{T:c2}, it is clear that $r_1$ 
verifies its statement; for 
Theorem~\ref{T:c4}, see the first part of its proof. 
Therefore, in the following we assume $r_1=10^{10}$ and take $c(1)=c(0)\cdot {M(10^{10})}$. 
Also we are taking $r_n$ big enough to satisfy simultaneously the conditions of Lemmas~\ref{C:ordencolumnas}, \ref{L:cocientec34}, Theorems~\ref{T:c2}, \ref{T:c4} and Corollaries~\ref{C:c1}, \ref{C:c5}. 
\end{remark}
In order to proceed to establish the linear independence of the second and fourth columns of $N_n$, 
as a previous step, it will be convenient to work with the angle $\alpha_{2,4}(1)$ {between the second and fourth columns} of the matrix $c(1)=c(0)\cdot M(10^{10})$.

\begin{lemma}\label{L:NuevoAlpha}
$\alpha_{2,4}(1)\approx  0.61287  $ radians, or $ 35.1150^{\mathrm{o}}$. 
\end{lemma}
\begin{proof}
With the help of a computer, we find
{\tiny
\begin{eqnarray*}
c_2(1)&=& ( 
6724585879762930009847217678925752031759280143132,4039352931172817287241960107428595442430070842597,\\
& & 351799079076971463272071897141315476273379662114,98096694645053186347651975943685219190151236056,\\
& & 13463673637462237721825378743371275851182965738,2168557055177101658034221170497937527466169936066,\\
& & 3781333487157436806125785995140120675027945513480,1735699798794056509112723701314448844530629244066,\\
& & 3416781840314699421838535355519580051273702574220,17177188800629772648590734203820777647998180299182
)
\end{eqnarray*}
}
 and 
{\tiny
\begin{eqnarray*}
c_4(1)&=& ( 26395487512479982089181921831124467308077144,15810254336235327591439949867884024917963530,\\
& & 1578029730137573497110630320282588346841714,110499169906965811457797221010650362296529025,\\
& & 271976802741742516476939281538365934964804,8589877719769222078334348720237110484879130,\\
& & 14814573942813803885778615438192476644697798,6809499482715096123851043378271526773523324,\\
& & 13408122709848376600064235351984483530361096,177959369951143463116119626469909395933953144
).
\end{eqnarray*}
}

Therefore, 
{\tiny
\begin{eqnarray*}
& & \cos(\alpha_{2,4}(1))=\frac{\left\langle c_2(1), c_4(1)\right\rangle}{\left\|c_2(1)\right\|_e \left\|c_4(1)\right\|_e} \\&=&\frac{ 3441880392395080399922695837058635750445825204823181953754814346118247133079596338663535984898 }{
\sqrt{  45348258594099938427599236739072860113897478329474707678837030240546527959551372587345105 }}\\
& &\times \frac{1}{\sqrt{  390413983851904440488080881284455556376331157805097083753494347951581607913506960663091586303699245 }}\\
&=&   0.8179990222644798...,
\end{eqnarray*}
}
so $\alpha_{2,4}(1)=\arccos(0.8179990222644798...)\approx 0.61287$ radians or $35.1150^{\mathrm{o}}$.
\end{proof}

\begin{lemma}\label{L:exponencial} Let $p_n=n+1, n\geq 1.$ Then
 $$\prod_{j=2}^{n}\left(1+10^{-p_j}\right)\leq e^{\frac1{900}}.$$
\end{lemma}
\begin{proof}
Observe that $$\prod_{j=2}^{n}\left(1+10^{-p_j}\right)\leq 
\prod_{j=2}^{n} e^{10^{-p_j}}=  e^{\sum_{j=2}^n 10^{-p_j}}
\leq e^{\frac{10^{-p_2}}{1-1/10}}=e^{^{\frac{10^{-3}}{9/10}}}= e^{^{\frac1{900}}}.$$
\end{proof}

\begin{proposition} The angle $\angulon$ satisfies 
$${34.91}^{\mathrm{o}}<\angulon<{35.42}^{\mathrm{o}}\,\textrm{ for any } n\in\N\cup\{0\},$$
{or, in radians, $$0.609 < \angulon <0.619.$$}
\end{proposition}
\begin{proof}
 We proceed by recurrence in $n$. For $n=0$ and $n=1$, see Lemma~\ref{L:valoresn=1} and Lemma~\ref{L:NuevoAlpha}, respectively. 
Suppose $n\geq 2$. We apply Theorem~\ref{T:c2} several times to obtain:
$$
2^{n-1}\rn r_{n-1}r_{n-2}\dots r_3 r_2 c_2(1)\leq 
\cdosn\leq 
2^{n-1}\rn r_{n-1}r_{n-2}\dots r_3 r_2 c_2(1)\prod_{j=2}^n\left(1+10^{{-p_j}}\right).$$
Reasoning in the same way, by using Theorem~\ref{T:c4} we obtain:
$$
\rn r_{n-1}r_{n-2}\dots r_3 r_2 c_4(1)\leq 
\ccuatron\leq 
\rn r_{n-1}r_{n-2}\dots r_3 r_2 c_4(1)\prod_{j=2}^n\left(1+10^{ {-p_j}}\right).$$
And applying Lemma~\ref{L:exponencial}:
\begin{equation}
\label{E:C} 
c_2(1)
\leq 
\frac\cdosn{2^{n-1}\rn r_{n-1}r_{n-2}\dots r_2}
\leq 
c_2(1)e^{\frac1{900}}
\end{equation}
and 
\begin{equation}
\label{E:D} 
c_4(1)\leq 
\frac\ccuatron{ \rn r_{n-1}r_{n-2}\dots r_2}\leq 
c_4(1)e^{\frac1{900}}.
\end{equation}
Observe now that the angle $\angulon$ equals the angle between the vectors
$\cdosn'=\frac\cdosn{K_2}$ and $\ccuatron'=\frac\ccuatron{K_4}$, with $K_2=2^{n-1}\rn r_{n-1}r_{n-2}\dots r_2$ and 
$K_4=\rn r_{n-1}r_{n-2}\dots r_2$. Put $\gamma:=\frac{\cdosn'}{c_2(1)}=(\gamma_1,\ldots,\gamma_{10})$ and 
$\beta{:=}\frac{\ccuatron'}{c_4(1)}=(\beta_1,\ldots,\beta_{10})$. Notice that 
$1\leq \gamma_j, \beta_j \leq e^{^{\frac{1}{900}}}$ \linero{and 
$\cdosn'=(\gamma_1 c_2(1)_1,\gamma_2 c_2(1)_2,\ldots, \gamma_{10}c_2(1)_{10}), 
\ccuatron'=(\beta_1 c_4(1)_1,\beta_2 c_4(1)_2,\ldots, \beta_{10}c_4(1)_{10}).$} Thus, using Equations~(\ref{E:C}) and (\ref{E:D}), {and Lemma~\ref{L:NuevoAlpha},} we have:
\begin{eqnarray*}
\cos\angulon&=&\frac{<\cdosn',\ccuatron'>}{\|\cdosn'\| \|\ccuatron'\| }\leq
\frac{\max_{1\leq j\leq 10}\{\gamma_j\}\max_{1\leq j\leq 10}\{\beta_j\}
<c_2(1),c_4(1)>}{\min_{1\leq j\leq 10}\{\gamma_j\}\min_{1\leq j\leq 10}\{\beta_j\}\|c_2(1)\| \|c_4(1)\|}\\
&\leq& \frac{e^{\frac1{900}}e^{\frac1{900}}}{1}\cos\alpha_{2,4}(1)\leq  0.818\cdot e^{^{\frac1{450}}} < 0.82.
\end{eqnarray*}
Similarly:
\begin{eqnarray*}
\cos\angulon&=&\frac{<\cdosn',\ccuatron'>}{\|\cdosn'\| \|\ccuatron'\| }\geq
\frac{\min_{1\leq j\leq 10}\{\gamma_j\}\min_{1\leq j\leq 10}\{\beta_j\}<c_2(1),c_4(1)>}
{\max_{1\leq j\leq 10}\{\gamma_j\}\max_{1\leq j\leq 10}\{\beta_j\}\|c_2(1)\| \|c_4(1)\|}\\
&\geq& \frac{1}{e^{\frac1{900}}e^{\frac1{900}}}\cos\alpha_{2,4}(1) \geq 0.817\cdot e^{^{-\frac1{450}}} > 0.815.
\end{eqnarray*}
Therefore $0.609385308030795...=\arccos(0.82) < \alpha_{2,4}(n) < \arccos(0.815)= 0.6180671318552149\ldots$.
Finally, we find  
$$34.91520^{\mathrm{o}}  < \alpha_{2,4}(n) <  35.41264^{\mathrm{o}}.$$

\end{proof}
%
%
%
%
%
%
%
%
%
%
%
%
%
%
%
%

\subsection{The limit direction of columns $2$ and $4$}
We know by the previous subsection that columns $2$ and $4$ of $N_n$ do not accumulate in the same direction. Our interest now is to show that both columns  have a limit direction. Recall that we have fixed the conditions on the choices of $r_n$ and $p_n$ in Remark~\ref{R:choicesrp}. 

In what follows, given $1\leq i,j\leq 10$, $n,m\in\N$ we will denote by $\delta_i(n,m)$ the angle between the $i$-th columns of $N_n$ and $N_m$; also $\alpha_{i,j}(n,m)$  will refer 
the angle between the $i$-th column of $N_n$ and the $j$-th column of  $N_m$.
{On the other hand, recall that $\lambda$ is the finite limit of a double sequence $f(m,n)$, $\lim_{m,n\rightarrow\infty}f(m,n)=\lambda,$ if and only if for any $\varepsilon>0$ there exists a positive integer $n_0=n_0(\varepsilon)$ such that 
$|f(m,n)-\lambda|<\varepsilon$ for all $m,n\geq n_0$.}

\begin{proposition}\label{P:4y5} It holds:
\begin{enumerate}
 \item[(1)] $\lim_{m,n\rightarrow\infty}\delta_4(m,n)=0$.
 \item[(2)] $\lim_{m,n\rightarrow\infty}\alpha_{4,5}(m,n)=0$.
\end{enumerate}

 \end{proposition}

\begin{proof}
 Let $n>m$  positive integers {(the reasoning for the case $n<m$ is supplied by the fact that $\delta_4(m,n)=\delta_4(n,m)$).} Use recursively Theorem~\ref{T:c4} to obtain:
 $$c_4(m)<\frac{\ccuatron}{r_{m+1}r_{m+2}\dots r_n}<c_4(m)\prod_{j=m+1}^n(1+10^{-p_j}).$$
 Write $c_4'(n):=\frac{\ccuatron}{r_{m+1}r_{m+2}\dots r_n}$ and $\beta(n,m):=\frac{c_4'(n)}{c_4(m)}$; then for any $1\leq i\leq 10$ we have (here, $p_j=j+1$ for $j\geq 1$):
 $$1<\beta(n,m)_i<\prod_{j=m+1}^n(1+10^{-p_j})<\prod_{j=m+1}^n e^{10^{-p_j}}=e^{\sum_{j=m+1}^n10^{-p_j}}
 \leq e^{\frac{10^{-p_{m+1}}}{9/10}}=e^{^{\frac1{9\cdot 10^{{m+1}}}}}.$$
 Now keep in mind that $\delta_4(n,m)$ equals the angle between 
 $c_4'(n)$ and $c_4(m)$, then
\begin{eqnarray*}
\cos\delta_4(n,m)&=&\frac{<c_4'(n),c_4(m)>}{\|c_4'(n)\|_e\|c_4(m)\|_e}
 \geq\frac{\min_{1\leq j\leq 10}\beta_j(n,m)<c_4(m),c_4(m)>}
 {\max_{1\leq j\leq 10}\beta_j(n,m)\|c_4(m)\|_e\|c_4(m)\|_e} \\
 &\geq&\frac{1}{e^{^{\frac1{9\cdot 10^{ {m+1}}}}}}=e^{^{-\frac1{9\cdot 10^{{m+1}}}}},
\end{eqnarray*}
and now, {given $\varepsilon>0$,} we can take $m$ big enough to have $\cos\delta_4(n,m)$ close enough to $1$ and then $\delta_4(n,m)<\varepsilon.$ {This proves $(1)$.}

Now, by  Corollary~\ref{C:c5} we have:
 $$c_4(m)<\frac{\ccincon}{r_{m+1}r_{m+2}\dots r_n}<c_4(m)\prod_{j=m+1}^n(1+10^{-p_j}).$$
Let $c_5'(n):=\frac{\ccincon}{r_{m+1}r_{m+2}\dots r_n}$ and $\gamma(n,m):=\frac{c_5'(n)}{c_4(m)}$, then reasoning as previously for $\beta(n,m)$ we obtain:
 $$1<\gamma(n,m)_i<
 e^{^{\frac1{9\cdot 10^{ {m+1}}}}}, \quad  1\leq i\leq 10.$$
Then the angle $\alpha_{{4,5}}(n,m)$  satisfies:
\begin{eqnarray*}
\cos\alpha_{ {4,5}}(n,m)&=&\frac{<c_5'(n),c_4(m)>}{\|c_5'(n)\|_e\|c_4(m)\|_e}
 \geq\frac{\min_{1\leq j\leq 10}\gamma_j(n,m)<c_4(m),c_4(m)>}
 {\max_{1\leq j\leq 10}\gamma_j(n,m)\|c_4(m)\|\|c_4(m)\|} 
 \geq
 e^{^{-\frac1{9\cdot 10^{{m+1}}}}}
\end{eqnarray*} 
and, similarly to the statement of $(1)$, we have $\alpha_{4,5}(n,m)<\varepsilon,$ {which proves $(2)$.}
\end{proof}

An easy consequence of last proposition is the existence of a common limit direction of columns $4$ and $5$.
\begin{coro} \label{C:convergencia45} The sequences $\left(\frac{c_4(n)}{\|c_4(n)\|}\right)_{n\in\N}$  and $\left(\frac{c_5(n)}{\|c_5(n)\|}\right)_{n\in\N}$ converge to a common limit.
\end{coro}

Next result is needed to prove below the existence of a common limit direction for the first and second column of $N_n$.

\begin{proposition} It holds: 
\begin{enumerate}
 \item[(1)] $\lim_{m,n\rightarrow\infty}\delta_2(n,m)=0$.
 \item[(2)] $\lim_{m,n\rightarrow\infty}\alpha_{2,1}(n,m)=0$. 
\end{enumerate}

\end{proposition}

\begin{proof}
The proof is completely analogous to that of Proposition~\ref{P:4y5}; now in Part (1) we use Theorem~\ref{T:c2}, and for Part (2) we apply Corollary~\ref{C:c1}. The details are left in charge of the reader. 
\end{proof}

An easy consequence of last proposition is, again, the existence of a (unique) limit direction for {columns $2$ and $1$} of $N_n$.

\begin{coro} \label{C:convergencia12}
The sequences $\left(\frac{c_2(n)}{\|c_2(n)\|}\right)_{n\in\N}$  and $\left(\frac{c_1(n)}{\|c_1(n)\|}\right)_{n\in\N}$ converge to a common limit.
\end{coro}

In the sequel, for any pair of vectors $\bbb,\sss\in\R^n$, the angle between them will be denoted by $\alpha_{\,\bbb,\,\sss}.$ 

\begin{lemma}\label{L:cocientebs} Let $\bbb$ and $\sss$ be in $\R_+^{{d}}$ and assume that  
$\left|\frac{\sss}{\bbb}\right|<K\in\R${. Then} $\cos(\alpha_{\bbb,\sss+\bbb}){ > }\frac1{1+K{\sqrt{d}}}.$

\end{lemma}

\begin{proof}
Since $\left|\frac\sss\bbb\right|<K$ then $s_j<K b_j$ for any $1\leq j\leq {d}$, and $|\sss|<K|\bbb|$. We claim that 
$\left\|\sss\right\|_e<K{\sqrt{d}}\left\|\bbb\right\|_e.$ Indeed, by the equivalence between norms (\ref{E:relnormas}), we know that $|\sss|\geq \frac{1}{ {\sqrt{d}}}\left\|\sss\right\|_e$ and 
$|\bbb|\leq \left\|b\right\|_e$; consequently, $\frac{1}{{\sqrt{d}}}\left\|\sss\right\|_e\leq |\sss|<K|\bbb|\leq K\left\|\bbb\right\|_e$, and thus $\left\|\sss\right\|_e<K{\sqrt{d}}\left\|\bbb\right\|_e$, as claimed.

 Observe that $  \left\|\bbb\right\|_e^2 = <\bbb,\bbb>\, \leq\, <\bbb+\sss,\bbb>=\left\|\bbb+\sss\right\|_e \left\| \bbb \right\|_e\cos(\as)\leq
 ( \left\|\bbb \right\|_e+ \left\|\sss \right\|_e) \left\|\bbb \right\|_e \cos(\as)$, thus
 $$\cos(\as)\geq\frac{\left\|\bbb\right\|_e}{ \left\|\bbb \right\|_e+\left\|\sss \right\|_e}= 
\frac{1}{\left(1+\frac{ \left\|\sss\right\|_e}{ \left\|\bbb\right\|_e}\right)} { > } 
 \frac{1}{1+K{\sqrt{d}}}.$$
\end{proof}

\def\angipi{\delta_{c_i(n+1), \mathcal{P}(n)}}

In what follows we denote by {$\mathcal{P}(n)$} the space  generated by $c_2(n)$, $c_4(n)$ and $c_5(n)$. We are going to show that the angle between {$\mathcal{P}(n)$} and $c_i(n+1)$, say {$\angipi$}, goes to 0 as $n$ goes to $\infty$,  for $i=1,\ldots,10$.

\begin{lemma} \label{L:convergenciapequegnas}
The following angles goes to $0$ as $n$ goes to $\infty$:
\begin{enumerate}
 \item The angle between $c_3(n+1)$  and $2c_1(n)+2c_2(n)+c_5(n)$.
 \item The angle between $c_6(n+1)$  and $2c_1(n)+2c_2(n)$.
\item The angle between $c_7(n+1)$ and $2\cunon+\cdosn$.
\item The angle between $c_9(n+1)$  and $3\cunon$.
\item The angle between $c_i(n+1)$, $i\in\{8,10\}$,  and $2\cunon$.
\end{enumerate}

%
\end{lemma}

\begin{proof} We prove (1) and (2). 
Let $\ssstresn=2\ctresn+2\cseisn+2\csieten+\cnueven+10\cdiezn$, 
  $\sssseisn=\ctresn+2\cseisn+2\csieten+\cnueven+8\cdiezn$,
  $\bbbtresn=2\cunon+2\cdosn+\ccincon$
  and 
  $\bbbseisn=2\cunon+2\cdosn${;} then by Claim~\ref{C:relationship} we have 
  $\ctresnMuno=\bbbtresn+\ssstresn$
  and $\cseisnMuno=\bbbseisn+\sssseisn$.
  
Observe that{, by Lemma~\ref{C:ordencolumnas},} $\sssin<17\ctresn$   and $\bbbin>4\ccincon$ 
for $i\in\{3,6\}$. Then  Lemma~\ref{L:cocientec34} yields 
 $\frac{\sssin}{\bbbin}<\frac{17\ctresn}{4\ccincon}<\frac{17}4 \frac1{20}10^{-p_n}${, $i\in\{3,6\}$.}
 Now we apply Lemma~\ref{L:cocientebs} to obtain that 
$\cos(\alpha_{\cinMuno,\bbbin})\geq \frac{1}{1+\frac{17}{80} 10^{-p_n}{\sqrt{10}}}$. Then $\cos(\alpha_{\cinMuno,\bbbin})$ goes to $1$ as $n$ goes to $\infty$ and the angle $\alpha_{\cinMuno,\bbbin}$ goes to $0$. This finishes the proof of (1) and (2). 

The proof of the other items follows with minor changes. {For instance, in $(3)$ take $s_7(n)=2c_7(n)+c_9(n)+4c_{10}(n)$ and $b_7(n)=2c_1(n)+c_2(n)$, with $c_7(n+1)=s_7(n)+b_7(n)$ by Claim~\ref{C:relationship}, and $\frac{s_7(n)}{b_7(n)}<\frac{7}{3}\frac{\csieten}{\ccuatron}<\frac73\frac1{20}10^{-p_n}$ by Lemmas~\ref{C:ordencolumnas} and
\ref{L:cocientec34}; for the rest of cases, proceed similarly.}
\end{proof}

As a consequence of  Corollaries~\ref{C:convergencia45}, \ref{C:convergencia12}   and 
Lemma~\ref{L:convergenciapequegnas} we have:
\begin{theorem}\label{T:limiteplano}
Let $c_2$ be the common limit point of $\left(\frac{\cunon}{|\cunon|}\right)_n$  and $\left(\frac{\cdosn}{|\cdosn|}\right)_n$ and let $c_4$ be the common limit point of $\left(\frac{\ccuatron}{|\ccuatron|}\right)_n$ and $\left(\frac{\ccincon}{|\ccincon|}\right)_n$.

 Then the  sequences $\left\{\left(\frac{c_i(n)}{|c_i(n)|}\right)_n\right\}_{i=1}^{10}$ converge to  ${\mathcal{P}}=\{{\alpha} c_2+{\beta} c_4: {\alpha}\geq 0, {\beta}\geq 0\}.$
\end{theorem}

\section{Proof of Main Theorem}\label{S:proofmaintheorem}

We begin by introducing some notation and a technical lemma. 
For a given real ${d}\times {d}$ matrix $A${,} we define $\Delta_A=A\Lambda^{{d}}${;} recall that
$\Lambda^{{d}}=\R_+^{{d}}=\{\sum_{{{i=1}}}^{{d}}\lambda_ie_i:\lambda_i>0\}${,} where $e_i$ denotes the $i$-th vector of the canonical basis of $\R^{{d}}$. Observe that 
$\Cl{\Lambda^d}=\{ \sum_{{{i=1}}}^{{d}}\lambda_ie_i:\lambda_i\geq 0\}$ and
$\Bd{\Lambda^d}=\{ \sum_{{{i=1}}}^{{d}}\lambda_ie_i:\lambda_i\geq 0 \textrm{ and at least one}{\,} \lambda_j=0, \, 1\leq j\leq {d} \}${, where $\Cl$ and $\Bd$ denote the closure and the boundary of a set of points of $\mathbb R^{d}$, respectively.}

\begin{lemma}\label{L:conomatriz} Let $A$, $B$ and $A_j$, $j\in\N$, be nonnegative invertible  ${d}\times  {d}$ matrices. For any $1\leq i\leq {d,}$ $a_i$ and $b_i$ denote the $i$-th columns of $A$ and $B${,} respectively. It holds:
\begin{enumerate}
 \item $\Delta_A=\{\sum_{i=1}^{{d}}\lambda_i a_i: \lambda_i>0\}=\{\sum_{i=1}^{{d}}\lambda_i \frac{a_i}{|a_i|}: \lambda_i>0\}$.
 \item  $\Delta_{AB}{\,\subseteq\,} \Delta_A$.
 \item If $B$ is positive{,} $\Cl \Delta_{AB}\backslash\{0\}\subsetneq \Delta_A${, and $\Delta_{AB}\,\subsetneq\,  \Delta_A$.}
 \item  $A\left(\cap_{j\in\N}A_1A_2\dots A_j\Lambda^{{d}}\right)=\cap_{j\in\N}AA_1A_2\dots A_j\Lambda^{{d}}${.}
 
\end{enumerate}
\end{lemma}

\begin{proof}
 We prove the first item. {Given $\lambda=(\lambda_i)\in\Lambda^{d}$, observe} that  
 ${A\lambda}=A\sum_{i=1}^{{d}}\lambda_ie_i=
 \sum_{i=1}^{{d}}\lambda_i Ae_i=\sum_{i=1}^{{d}}\lambda_i a_i$ and the first  equality holds. The second equality is trivial. 
 Note now that{,} since $B$ is nonnegative and invertible{,} $B\Lambda^{{d}}\subseteq \Lambda^{{d}}$ and then $AB\Lambda^{{d}}{\subseteq} A\Lambda^{{d}}$, thus  (2) holds. 
 We now prove the third item. Let $u\in\Cl \Delta_{AB}\backslash\{0\}=(\Cl AB\Lambda^{{d}})\backslash\{0\}${;} then there exists $(\lambda_i)\in(\Cl \Lambda^{{d}})\backslash\{0\}$ such that $u=AB(\lambda_i)=\sum_{i=1}^{{d}} A B\lambda_i e_i=\sum_{i=1}^{{d}} A \lambda_i b_i=A\sum_{i=1}^{{d}}\lambda_i b_i$. Since  
  $(\lambda_i)\not =0$ there exists 
  $j\in\{1,2,\dots, {d}\}$ such that 
  $\lambda_j\not=0${;} also $b_j\not=0$ by hypothesis and then 
$\sum_{i=1}^{{d}}\lambda_i b_i>0$, therefore $u\in \Delta_A$. Moreover the equality does not hold since {$\Delta_{A}$} is open and 
$\Cl \Delta_{AB}\backslash\{0\}$  is not. {As a direct consequence, $\Delta_{AB}\,\subsetneq\,  \Delta_A$ because $\Cl\Delta_{AB}\setminus\{0\}\supsetneq \Delta_{AB}\backslash\{0\}=\Delta_{AB}$.} 

Finally we prove (4). Let us first see ``$\subseteq$''{;} take $u\in A\left(\cap_{j\in\N}A_1A_2\dots A_j\Lambda^{{d}}\right)${,} then there exists a sequence $(\lambda^j)_{j=1}^\infty$, $\lambda^j\in\Lambda^{{d}}$ for any $j\in\N$,  such that $u=AA_1A_2\dots A_j\lambda^j${;} then  
$u\in \cap_{j\in\N}AA_1A_2\dots A_j\Lambda^{{d}}$. Now we show  ``$\supseteq$''{;} let $u\in \cap_{j\in\N}AA_1A_2\dots A_j\Lambda^{{d}}{,}$ then there exists a sequence $(\lambda^j)_{j=1}^\infty$, $\lambda^j\in\Lambda^{{d}}$ for any $j\in\N$,  such that $u=AA_1A_2\dots A_j\lambda^j$, then $A^{-1}u=A_1A_2\dots A_j\lambda^j$ and 
$A^{-1}u\in \cap_{j\in\N} A_1A_2\dots A_j\Lambda^{{d}}{,}$ thus $u\in A {(}\cap_{j\in\N} A_1A_2\dots A_j\Lambda^{{d}}{)}$. 
\end{proof}

\subsection{Proof of Theorem~\ref{T:xxlnoperiodico}}

For any $k\in\N$ we write {
$$P_k:= M_{v_{n_k}}(\pi_{n_k})\cdot
M_{v_{n_k+1}}(\pi_{n_k+1})\cdot \ldots \cdot M_{v_{n_{k+1}-1}}(\pi_{v_{n_{k+1}-1}})$$} and 
{$$\Delta_k:=M_{v_1}(\pi_1)\cdot
M_{v_2}(\pi_2)\cdot \ldots \cdot M_{v_k}(\pi_k) \Lambda^{{d}}.$$}
Note that by Lemma~\ref{L:conomatriz}-(2) $\Delta_k{\supset} \Delta_{k+1}$ for any $k\in\N$ and 
$\bigcap_{i=1}^\infty \Delta_i=
\bigcap_{i=1}^\infty 
\Delta_{n_i}.
$ Thus, the following equality is true:
\begin{equation}
\bigcap_{i=1}^\infty M_{v_1}(\pi_1)\cdot
M_{v_2}(\pi_2)\cdot \ldots \cdot M_{v_i}(\pi_i) \Lambda^{{d}}=
\bigcap_{i=1}^\infty 
P_1\cdot P_2
\cdot \ldots \cdot P_{i}\Lambda^{{d}}.
\label{E:conoreducido} 
\end{equation}

Let us proceed with the first item in Theorem~\ref{T:xxlnoperiodico}{;} it will be enough if we show that   
$ \bigcap_{i=1}^\infty 
P_1\cdot P_2
\cdot \ldots \cdot P_{i}\Lambda^{{d}}=\bigcap_{i=1}^\infty 
\Delta_{{n_{_{i+1}}-1}}$ is nonempty.
Observe that{, by applying Lemma~\ref{L:conomatriz} to the positive matrices $P_i$, with $n_1=1$, $n_{k+1}>n_k$,}
$$\Cl\Delta_{{n_{_{i+1}}-1}}\backslash\{0\}\supsetneq \Delta_{{n_{_{i+1}}-1}}\supsetneq\Cl\Delta_{{n_{_{i+2}}-1}}\backslash\{0\}\supsetneq \Delta_{{n_{_{i+2}}-1}},$$
and 
$$\Cl\Delta_{{n_{_{i+1}}-1}}\cap \S^{{d}-1}\supsetneq 
\Delta_{{n_{_{i+1}}-1}}\cap \S^{ {d}-1}
\supsetneq
\Cl\Delta_{ {n_{_{i+2}}-1}} \cap \S^{ {d}-1}
\supsetneq 
\Delta_{ {n_{_{i+2}}-1}}\cap \S^{{d}-1},$$
{where $\S^{d-1}$ denotes the $(d-1)$-sphere.}
Then 
$C=\bigcap_{i=1}^\infty\left( \Delta_{{n_{_{i+1}}-1}}\cap \S^{{d}-1}\right)=\bigcap_{i=1}^\infty \left(\Cl\Delta_{{n_{_{i+1}}-1}}\cap\S^{{d}-1} \right)$ 
is a nonempty compact set (cf.~\cite[Th. {5.9}, Ch. 3]{munkres}) and also $\bigcap_{i=1}^\infty \Delta_{{n_{_{i+1}}-1}}$ is nonempty. 
The second item follows by applying \cite[Lemma 18]{linerosolerergodic} repeteadly to any 
$\lambda^1\in{\mathcal{C} (}\mathcal{G}^{\pi_1,v}{)}$.
Item (3) follows from \cite{mioangosto} and (4)  from Theorem~\ref{T-cone}-(3) and Lemma~\ref{L:conomatriz}-(4). {$\hfill\Box$}

\subsection{Finishing the proof {of Main Theorem}}

Let $u\in\{a,b\}^{\N}$ be the vector defined by (\ref{E:chosenpath}) in Section~\ref{Subs:path}, let $$\pi_0=(-3,-4,-5,-6,-7,-8,-9,10,1,-2){,}$$ and 
consider  the Rauzy-subgraph 
associated to $\pi_0$ and $u$. 
Recall that 
$u=v^0*w^1*w^2*w^3*\dots*w^k*\dots$, where $w^j=v(r_j,r_j)$, see~(\ref{eqV}), with a suitable  increasing sequence $(r_j)_j$ satisfying the properties in Remark~\ref{R:choicesrp}. Note  that 
$M_{v^0}(\pi_0)$ is positive, see (\ref{E:c0}), 
$v^0(\pi_0)=\pi_0$ and 
$v^0*w^1*w^2*w^3*\dots*w^k(\pi_0)=\pi_0$ for any $k\in\N$, cf. Proposition~\ref{P:ciclo}.
In what follows, given two real ${d}\times {d}$ matrices $A=(a_{i,j})$ and $B=(b_{i,j})$, we write $A\geq B$ if $a_{i,j}\geq b_{i,j}$ for any $1\leq i,j\leq {d}.$

Take into account Equation~(\ref{Eq:M(r)}) {(in particular, $M_{w^{k}}(\pi_0)=M(r_k)$, with 
$M(r_k)\geq M(1)$);} {since each matrix $M_{w^{5j+s}}(\pi_0)$, $s\in\{1,\ldots,5\},$ has the same distribution of zeros as $M(1)$ and they differ only in the entries containing the values $r_{5j+s}$,}  for any $j\in\N$ {we have:}
\begin{eqnarray*}
M_{w^{5j+1}}(\pi_0)M_{w^{5j+2}}(\pi_0)M_{w^{5j+3}}(\pi_0)M_{w^{5j+4}}(\pi_0)M_{w^{5j+5}}(\pi_0)\geq M(1)^5= \\
\begin{pmatrix}
39272&64132&72637&119636&107632&53638&25558&13266&20228&
 14654\\ 17386&28410&32290&53432&48030&23780&11295&5846&8912&6446
 \\ 1059&1753&2589&5505&4749&1601&655&337&513&368\\ 30&56&470&1691&
 1370&130&8&4&6&4\\ 129&223&697&2095&1729&283&65&33&50&35\\ 8820&
 14426&16744&28410&25423&12154&5710&2956&4506&3258\\ 17386&28409&
 32290&53432&48030&23780&11296&5846&8912&6446\\ 10800&17633&19929&
 32735&29465&14739&7033&3654&5572&4039\\ 20596&33630&38050&62586&
 56320&28119&13408&6962&10616&7692\\ 84081&137408&157716&264200&
 236962&115365&54586&28288&43127&31212\\ 
\end{pmatrix}{.}
\end{eqnarray*}
Then we can apply Theorem~\ref{T:xxlnoperiodico}{, with $n_k=5k+1$, and $$P_k=M_{w^{5k+1}}(\pi_0)M_{w^{5k+2}}(\pi_0)M_{w^{5k+3}}(\pi_0)M_{w^{5k+4}}(\pi_0)M_{w^{5k+5}}(\pi_0),$$}to obtain the existence of a minimal flipped \iet\, $T=(\lambda,\pi_0)$ 
whose associated graph is $\mathcal{G}^{\pi_0,u}$. Moreover{, by the same theorem,} we also find minimal flipped \iets\ with the associated permutations  listed 
in Lemma~\ref{L:permutations}{;} a simple inspection shows the existence of minimal $(10,k)$-flipped \iets\ with $2\leq k\leq 9$ which will be 
non uniquely ergodic since, as we will show, $\mathcal{C}(\mathcal{G}^T)$ has {dimension} 2 and Theorem~\ref{T-cone}-(3) gives the non unique ergodicity.
The existence of minimal flipped non uniquely ergodic $(10,10)$ and $(10,1)$-\iets\ is an easy consequence of Lemma~\ref{L:conomatriz}-(4) and the following relations
in the Rauzy graph (see again Lemma~\ref{L:permutations}):
$$a*b((-2,-3,-4,-5,-6,-7,-8,-10,-1,-9))=\pi_0,\quad
b*a((9,1,10,8,7,6,5,4,3,-2))=\pi_{11}{,}$$
{where notice that the operations are applied from right to left.}

 \def\conoN{\bigcap_{i=0}^\infty N_i \Lambda^{{d}}}

We finish the proof {of Main Theorem by showing that the dimension of $\mathcal{C}(\mathcal{G}^T)$ is just $2$, therefore by Theorem~\ref{T-cone}-(3) we will deduce the non unique ergodicity of $T$}. {To this end,} {r}easoning as in Equation~(\ref{E:conoreducido}) we have:
\begin{equation}\label{Eq:neo}
\bigcap_{i=1}^\infty M_{u_1}(\pi_0)\cdot
M_{u_2}(\pi_1)\cdot \ldots \cdot M_{u_i}(\pi_{i-1}) \Lambda^{ {d}}=
\conoN{,}
\end{equation} 
{where $N_i=M_{v^0*w^1*\ldots*w^i}(\pi_0),$} and then we will be done if we show \begin{equation}
{\left(\bigcap_{n=0}^\infty N_n \Lambda^{d}\right)\setminus\{0\}=\mathcal{P}\setminus\{0\},} 
\label{E:conoigualplano}                                     
                                    \end{equation}
{being $\mathcal{P}=\{\alpha c_2+\beta c_4:\alpha,\beta\geq 0\},$ and $c_2, c_4$ are the vectors given in Theorem~\ref{T:limiteplano}. }

                                    First we begin with the inclusion $\subseteq$. Let ${z}\in
{\bigcap_{n=0}^\infty N_n \Lambda^{d},}$ then 
there exist sequences {of positive real numbers $(\lambda^i_n)_{n\in\N}$,}  $1\leq i\leq 10={d}$,  such that
${z}=\sum_{i=1}^{10}\lambda^i_n \frac{c_i(n)}{|c_i(n)|}$.
Observe that the sequences $(\lambda^i_n)_{n\in\N}$ are bounded{, since $z$ is fixed and each sequence $\left(\frac{c_i(n)}{|c_i(n)|}\right)$ converges by Theorem~\ref{T:limiteplano}; then}  there exists a sequence of naturals $(n_k)_k$ such that $(\lambda^i_{n_k})_k$ converges to $\lambda^i\in\R${, $1\leq i\leq 10$.}

Observe that the sequence $({z}_k)_k$, ${z}_k=\sum_{i=1}^{10}\lambda^i_{n_k} \frac{c_i(n_k)}{|c_i(n_k)|}={z,}$ converges to ${z}$ and {according to Theorem~\ref{T:limiteplano}}   ${z}\in{\mathcal{P}}$.

Now we prove $\supseteq$. Let ${z}\in{\mathcal{P}}$, {$z\neq0$;} then  ${z}=\lambda^2c_2+\lambda^4c_4$ {for some $\lambda^2\geq 0, \lambda^4\geq 0, \lambda^2+\lambda^4>0.$} 
Denote by $(n_i)_i$ the strictly increasing sequence of naturals such that {(see~(\ref{Eq:neo}))} 
$M_{u_1}(\pi_0)\cdot M_{u_2}(\pi_1)\cdot \ldots \cdot M_{u_{n_i}}(\pi_{n_i-1})=N_i{,}$ whose column are $c_j(i)$. Then by Lemma~\ref{L:conomatriz}: 
$$\Delta_i=\left\{\sum_{l=1}^{10} \lambda_l c_l(i): \lambda_l>0\right\}, \quad\Cl \Delta_i=\left\{\sum_{l=1}^{10} \lambda_l c_l(i): \lambda_l\geq 0\right\}
$$
 and
 $$\Cl\Delta_i{\setminus\{0\}}\supsetneq \Delta_i\supsetneq\Cl\Delta_{i+1}\backslash\{0\}\supsetneq \Delta_{i+1},\quad i\in\N. $$
{ We are going to prove that $z\in\Cl\Delta_j\setminus\{0\}$ for all $j\geq 0$. Given $\varepsilon >0,$ take $\widetilde{\varepsilon}=\min\{\varepsilon,\left\|z-0\right\|_{e}/2\}.$ By Theorem~\ref{T:limiteplano}, and ${z}\in{\mathcal{P}}$, there exists $n_0=n_0(\varepsilon)$ such that $\left\|\lambda^2\frac{c_2(n)}{\left|c_2(n)\right|}+\lambda^4 \frac{c_4(n)}{\left|c_4(n)\right|} -z\right\|<\widetilde{\varepsilon}$ for all $n\geq n_0.$ Since $0\notin B_{\widetilde{\varepsilon}}(z)$ (the ball of radius $\widetilde{\varepsilon}$ and center $z$), we deduce that $B_{\widetilde{\varepsilon}}(z)\cap \Cl\Delta_j\setminus\{0\}\neq\emptyset$ for all $j\geq n_0$. Automatically, by the nest structure of $\Delta_j$'s, we infer that $B_{\widetilde{\varepsilon}}(z)\cap \Cl\Delta_j\setminus\{0\}\neq\emptyset$ for all $j\geq0,$ or even 
$B_{\widetilde{\varepsilon}}(z)\cap \Cl\Delta_j\neq\emptyset$ for all $j\geq 0$. By definition of closure, since $\varepsilon$ was arbitrarily taken, we deduce that $z\in\Cl\Delta_j$ for all $j\geq 0$, and therefore $z\in\cap_{j\geq 0}\Cl\Delta_j$. Since $z\neq 0$, we finally find $z\in\left(\cap_{j\geq 0}\Cl\Delta_j\right)\setminus\{0\}=\cap_{j\geq 0}\left(\Cl\Delta_j\setminus\{0\}\right)=\cap_{j\geq 0}\Delta_j$, that is, $z\in \bigcap_{j=0}^\infty N_j \Lambda^{d}$, $z\neq 0$, as desired. 
}
{$\hfill\Box$}


%
%
%
%

\subsection{A consequence on transitive IETs}
We can use Main Theorem to construct transitive non uniquely ergodic $(n,k)$-IETs, for $n\geq 12$ and $1\leq k\leq n$, if $n$ is even, and $1\leq k\leq n-1$ if $n$ is odd. To this end, 
we need some preliminaries about measures. Before, let us recall the construction appearing in \cite{gutierrez4b} in order to obtain, from a transitive $(n,f)$-\iet\, $T:D\subset [0,1]\to[0,1]$, two new transitive IETs, namely, the transitive $(n+1,f)$-\iet,
$T_1:D_1\subset [0,2]\to[0,2]$, and the transitive $(n+2,f+2)$-\iet, $T_2:D_2\subset [0,3]\to[0,3],$ given by:

\begin{equation}\label{Def:T1}
T_1(x)=\left\{\begin{array}{ll}
T(x)+1&\textrm{ if }x\in D\cap [0,1],\\
x-1&\textrm{ if }x\in(1,2),
\end{array}\right. \quad \qquad
T_2(x)=\left\{\begin{array}{ll}
T(x)+1&\textrm{ if }x\in D\cap [0,1],\\
-x+4&\textrm{ if }x\in(1,2),\\
-x+3&\textrm{ if }x\in(2,3).\\
\end{array}\right.
\end{equation}
In the next results we analyze the preimages of $T_1$ and $T_2$. Given a positive integer $k$ and an arbitrary set $X$, 
we put $X+k=k+X=\{x+k:x\in X\}$, $X-k=\{x-k:x\in X\}$, and $k-X=\{k-x:x\in X\}.$

\begin{lemma}\label{L:inversasT1} The following properties hold:
\begin{itemize}
\item[(a)] $T_1^{-1}(A)=A+1,$ if $A\subseteq (0,1);$
\item[(b)] $T_1^{-1}(B)=T^{-1}(B-1),$ if $B\subseteq (1,2);$
\item[(c)] $T_1^{-1}(C)=\left(\left[C\cap(0,1)\right]+1\right) \cup T^{-1}\left(\left[C\cap (1,2)\right]-1\right).$
\end{itemize}
\end{lemma}
\begin{proof}
(a) Notice that $T_1^{-1}(A)\subset [1,2]$. Then, $z\in T_1^{-1}(A)$ iff $T_1(z)=z-1\in A$, and from here it is easily seen that $T_1^{-1}(A)=A+1$. 

(b) In this case, $T_1^{-1}(B)\subset [0,1]$, and $z\in T_1^{-1}(B)$ iff $T_1(z)=T(z)+1\in B$, and we derive $T_1^{-1}(B)=T^{-1}(B-1).$

(c) It is an immediate consequence of (a)-(b) and the fact that $T_1^{-1}$ preserves the union of subsets.
\end{proof}

Following with the same strategy developed for $T_1$, we present the following result about the preimages of $T_2^{-1}$. 
\begin{lemma}\label{L:inversasT2} It holds:
\begin{itemize}
\item[(a)] $T_2^{-1}(A)=3-A,$ if $A\subseteq (0,1);$
\item[(b)] $T_2^{-1}(B)=T^{-1}(B-1),$ if $B\subseteq (1,2);$
\item[(c)] $T_2^{-1}(C)=4-C$, if $C\subseteq (2,3).$
\end{itemize}
\end{lemma}
\begin{proof}
(a) Notice that $T_2^{-1}(A)\subset [2,3].$ Moreover, if $a\in A$, we find $T_2(3-a)=-(3-a)+3=a,$ and the injectivity of $T_2$ finishes this case.

(b) Now, $T_2^{-1}(B)\subset [0,1].$ Since $z\in T_2^{-1}(B)$ iff $T_2(z)=T(z)+1\in B,$ we easily deduce the property on preimages when $B\subset (1,2).$

(c) Realize that $T_2^{-1}(C)\subset [1,2]$ and take into account that $T_2(4-c)=-(4-c)+4=c$.
\end{proof}

Let $\mu$ be an invariant probability measure associated to $T$. We define a new measure $\mu_1$ on the Borelians of $(0,1)\cup (1,2)$ in this manner:
 \begin{equation}\label{Def:mu1}
\mu_1(x)=\left\{\begin{array}{ll}
\frac{1}{2}\mu(A),&\textrm{ if }A \subseteq (0,1),\\
\frac{1}{2}\mu(A-1),&\textrm{ if }A \subseteq (1,2),\\
\frac{1}{2}\mu(C\cap(0,1))+\frac{1}{2}\mu\left( \left[C\cap(1,2)\right]-1\right),&\textrm{ in\, the\, general\,case.}
\end{array}\right.
\end{equation}

\begin{proposition}\label{P:paraT1}
If $\mu$ is an invariant  probability measure for $T$, then $\mu_1$  is so for $T_1$. 
\end{proposition}
\begin{proof}
It is a simple matter to check that $\mu_1$ is a probability measure on the Borel sets of $(0,1)\cup (1,2)$. 
We now prove that $\mu_1$ is invariant for $T_1$, that is, $\mu_1(T_1^{-1}(C))=\mu_1(C)$ for all $C.$ By Lemma~\ref{L:inversasT1}, 
\begin{eqnarray*}
\mu_1(T_1^{-1}(C))&=&\mu_1(\left(\left[C\cap (0,1)\right]+1\right) \cup T^{-1}\left(\left[C\cap (1,2)\right]-1\right))\\
&=&\mu_1(\left(\left[C\cap (0,1)\right]+1\right))+ \mu_1(T^{-1}\left(\left[C\cap (1,2)\right]-1\right))\\
&=&\frac{1}{2}\mu(\left(\left[C\cap (0,1)\right]+1\right)-1)+\frac{1}{2}\mu(T^{-1}\left(\left[C\cap (1,2)\right]-1\right))\\
&=&\frac{1}{2}\mu(C\cap (0,1))+\frac{1}{2}\mu(\left(C\cap(1,2)\right)-1)=\mu_1(C),
\end{eqnarray*}
where in the last line we use the invariance of $\mu$ and the definition of $\mu_1$. 
\end{proof}

In the next step, we will associate to $T_2(x)$ a new measure $\mu_2$. We define: 
\begin{equation}\label{Def:mu2}
\mu_2(x)=\left\{\begin{array}{ll}
\frac{1}{3}\mu(A),&\textrm{ if } A \subseteq (0,1),\\
\frac{1}{3}\mu(A-1),&\textrm{ if } A \subseteq (1,2),\\
\frac{1}{3}\mu(3-A),&\textrm{ if } A \subseteq (2,3),\\
\frac{1}{3}\mu(A\cap (0,1))+\frac{1}{3}\mu\left( \left[A\cap(1,2)\right]-1\right)+\frac{1}{3}\mu(3-\left(A\cap(2,3)\right)),&\textrm{ in\, the\, general\,case.}
\end{array}\right.
\end{equation}
It is immediate to check that $\mu_2$ is a probability measure, and we left to the reader in charge of the proof. In fact:

\begin{proposition}\label{P:paraT2}
If $\mu$ is an invariant probability measure for $T$, then $\mu_2$ is so for $T_2$. 
\end{proposition}
\begin{proof}
We will be done by proving that $\mu_2$ is invariant, that is, $\mu_2(T_2^{-1}(X))=\mu_2(X)$ for all $X\subseteq [0,3]$. We apply the fact that $T_2^{-1}$ preserves the unions, we use Lemma~\ref{L:inversasT2}, we consider that $\mu$ is invariant for $T$,  and we take into account the definition of $\mu_2$:
\begin{eqnarray*}
\mu_2(T_2^{-1}(X))&=&\mu_2(T_2^{-1}(X\cap (0,1)) \cup T_2^{-1}(X\cap(1,2)) \cup T_2^{-1}(X\cap(2,3)))\\
&=&\mu_2(3-(X\cap (0,1)))+\mu_2(T^{-1}(\left(X\cap (1,2)\right)-1))+\mu_2(4-(X\cap (2,3)))\\
&=&\frac{1}{3}\mu(3-\left(3-(X\cap (0,1))\right)) +\frac{1}{3}\mu\left(T^{-1}(\left(X\cap (1,2)\right)-1)\right)
+\frac{1}{3}\mu\left( (4-(X\cap (2,3)))-1\right)\\
&=& \frac{1}{3}\mu(X\cap (0,1)) +\frac{1}{3}\mu\left(\left(X\cap (1,2)\right)-1\right)
+\frac{1}{3}\mu\left( 3-(X\cap (2,3))\right)=\mu_2(X).
\end{eqnarray*}
\end{proof}
\subsubsection{Proof of Corollary~\ref{C:transitivo}}

\begin{figure}[thb]
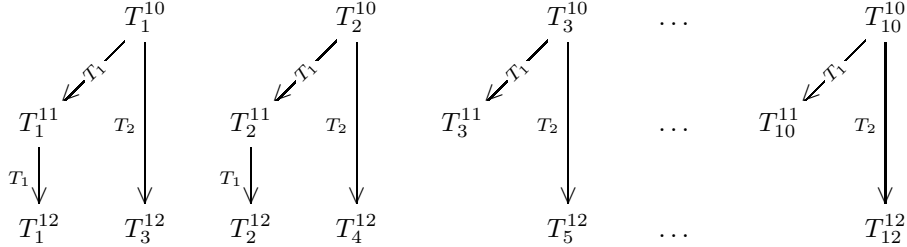

$$\begin{diagram}
&                &T^{10}_1   &&        &&T^{10}_2&  &&  &T^{10}_3&&\ldots&&&&T^{10}_{10}&&&\\
&\ldTo(2,2)~{T_1}& & &     &\ldTo(2,2)~{T_1}&&&&\ldTo(2,2)~{T_1}&&&&&&\ldTo(2,2)~{T_1}&&&&&&&&&&&\\
T^{11}_1&&\dTo^{T_2}           &&T^{11}_2&&\dTo^{T_2}&&T^{11}_3&&\dTo^{T_2}&&\ldots&&T^{11}_{10}&&\dTo^{T_2}&&&&\\
\dTo^{T_1}&&           &&\dTo^{T_1}&&&&&&&&&&&&&& &&&&&&&&&\\
T^{12}_1 &&T^{12}_3     &&T^{12}_2&&T^{12}_4&&&&T^{12}_5&&\ldots&&&&T^{12}_{12}&&&&&&&&&&&\\
\end{diagram}$$
\caption{Graph of transitive \iets\  generated by means of $T_1$ and $T_2$ \label{tree}}
\end{figure}

\begin{figure}[thb]
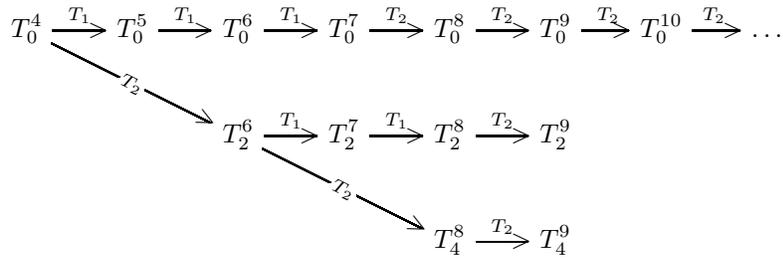

$$\begin{diagram}
T^4_0&   \rTo^{T_1} &T^{5}_0&   \rTo^{T_1} &T^6_0&  \rTo^{T_1} &T^{7}_0&   \rTo^{T_2} &T^{8}_0 &\rTo^{T_2} &T^{9}_0 &\rTo^{T_2} &T^{10}_0 &\rTo^{T_2} &\ldots \\
&\rdTo(4,2)~{T_2}&& &&  &&    & & && && &\\
&&& &T^{6}_2&   \rTo^{T_1} &T^7_2&  \rTo^{T_1} &T^{8}_2&   \rTo^{T_2} &T^{9}_2 &&& &  \\
&&& &&\rdTo(4,2)~{T_2}  &&    & & && && & \\
&&& &&   &&   &T^{8}_4&   \rTo^{T_2} &T^{9}_4 &&& &  \\
\end{diagram}$$
\caption{Transitive \iets\  generated by means of $T_1$ and $T_2$ beginning with the minimal non uniquely ergodic \iet\ by Keane \label{tree2}}
\end{figure}

We are going to prove the existence of  transitive non uniquely ergodic $(n,k)$-\iets\, for all $n\geq 10$ and $1\leq k\leq n$ if $n$ is even, and  $1\leq k\leq n-1$ if $n$ is odd. 
By Main Theorem, we get minimal non uniquely ergodic proper $(10,k)$-\iets\, for $1\leq k\leq 10$, denoted in Figure~\ref{tree} by $T^{10}_k$. In particular, these maps are transitive. By the construction of maps $T_1$ and $T_2$, see~(\ref{Def:T1}) and Figure~\ref{tree}, we can create new transitive \iets: through the application 
of $T_1$ we obtain new transitive $(11,k)$-\iets\, that we denote by $T^{11}_k$, for $1\leq k\leq 10$; and $T_2$ provides transitive $(12,k+2)$-\iets, 
$3\leq k+2\leq 12$ (we denote them as $T^{12}_{k+2}$). To obtain transitive $(12,1)$- and $(12,2)$-\iets, apply the map $T_1$ to   
$T^{11}_1$ and $T^{11}_2$, respectively. In this way, we find transitive $(12,k)$-\iets\, for $1\leq k \leq 12.$ Additionally, by Propositions~\ref{P:paraT1}-\ref{P:paraT2}, we deduce that $T^{12}_k$ are non uniquely ergodic, $1\leq k\leq 12$, otherwise we would obtain that $T^{10}_k$ would not be non uniquely ergodic, in contradiction with our choice of $T^{10}_k$. 
Repeating the procedure it is a simple task to conclude the existence of transitive non uniquely ergodic $(n,k)$-\iets\, for all $n\geq 10$ and $1\leq k\leq n$ if $n$ is even. For the case $n$ odd, following a similar procedure, we can ensure the existence of transitive non uniquely ergodic \iets\, having $k$ flips, but with the restriction on $k$ given by $1\leq k <n$.  

The second statement of the corollary  follows with the same arguments and  by using the operators $T_1$ and $T_2$  in the scheme shown in Figure~\ref{tree2}, where $T^4_0$ is the minimal oriented non uniquely ergodic  $4$-\iet\ built by Keane in \cite{keanenounicamenteergodicas}, $T^n_0$ denotes a proper oriented $n$-\iet, $n\in\N$, and $T^n_k$ a proper $(n,k)$-\iet\ for  naturals $n,k$.

$\hfill\Box$

\section{Conclusions and further directions} 
Our Main Theorem highlights the existence of minimal non uniquely ergodic proper $(10,k)$-\iets\ and as a consequence we have deduced the existence of 
several types of transitive non uniquely ergodic $(n,k)$-\iets, see Corollary~\ref{C:transitivo}. However,  we stress that we have only built examples with two independent invariant measures. 
Then we propose to analyze, in future works, some  problems.

First of all, it would be interesting to get a bound for the number, $N$, of independent invariant measures (the dimension of the cone of invariant mesures) that an $(n,k)$-\iet\ can admit when we know its associated permutation. This problem is solved in the oriented case. Indeed, if  $T$ is a  minimal oriented \iet,  then M. A. Veech shows,  in \cite[Th. 2.12]{veech},  that $N\leq \frac{n}{2}$, in fact the bound may be sharper. He proves that $N\leq \frac{R}{2}$ where $R$ is the rang of the $n\times n$ translation matrix associated to $T$ which eventually is not $n$. Also, in \cite[S. IV]{YoccozCurso}, the reader can follow an analysis in the following terms: $N$ is bounded by $g$, which is the genus of the suspension surface associated to $T$ and this genus satisfies $g=1+\frac{n-m}{2}$, where $m$ is the number of the so called marked points of the surface, see \cite[S. 5-6]{espinperaltasoler} for more details. Then we propose to give bounds, in similar terms, for the flipped case.

Secondly, we think that following the technic used in the present work, it would be interesting to construct other minimal $(n,k)$-\iets. While this seems feasible for $n\geq 11$ following the Rauzy-graphs from \cite{linerosolerergodic} by repeating some fixed vertices by either the operator $a$ or $b$, we do not know if it will be possible for $n\leq 9$. In any case it would be an attractive problem to determine the minimal integer $n$ for which there exists a minimal non uniquely ergodic $(n,k)$-\iet. In the oriented case this bound is $4$ by the mentioned work of Veech and this bound is realized by the example of Keane, \cite{keanenounicamenteergodicas}.

Finally, Corollary~\ref{C:transitivo} shows the existence of transitive non uniquely ergodic $(n,k)$-\iets, but there are some gaps in the statements. For example the cases $(2n+1,2n+1)$ for any integer $n$. We encourage to fulfill these gaps.

%


\section{Acknowledgements}
This paper has been partially supported by Grant  MTM2017-84079-P (AEI/FEDER, UE), Ministerio de Ciencias, Innovaci\'{on} y Universidades,
Spain. 

This paper has been supported by the grant number MTM2014-52920\_p  from Ministerio de Econom\'ia y 
Competitividad (Spain).


\def\cprime{$'$}

\end{document}